\documentclass[11pt]{article}
\usepackage{booktabs} % For formal tables
\usepackage{fullpage}
\usepackage{hyperref}
\usepackage{url}
\usepackage{natbib}

\usepackage[ruled]{algorithm2e} % For algorithms

\SetAlFnt{\small}
\SetAlCapFnt{\small}
\SetAlCapNameFnt{\small}
\SetAlCapHSkip{0pt}
\IncMargin{-\parindent}

\usepackage[utf8]{inputenc}
\usepackage{xcolor}
\usepackage{amsmath,amsthm,amsfonts,amssymb,microtype,bbm}
\usepackage{graphicx}

\theoremstyle{plain}
\newtheorem{theorem}{Theorem}
\newtheorem{definition}{Definition} 
\newtheorem{example}{Example} 
\newtheorem{lemma}{Lemma}

\newtheorem{corollary}{Corollary}
\newtheorem{claim}{Claim}
\newtheorem{assumption}{Assumption}
\newtheorem{remark}{Remark}

\DeclareMathOperator*{\argmax}{arg\,max}

\newcommand{\reals}{\mathbb{R}}
\newcommand{\cD}{\mathcal{D}}
\newcommand{\cA}{\mathcal{A}}

\newcommand{\E}{\mathbb{E}}
\newcommand{\cM}{\mathcal{M}}
\newcommand{\cF}{\mathcal{F}}
\newcommand{\cB}{\mathcal{B}}
\newcommand{\cI}{\mathcal{I}}
 
\newcommand{\nph}{\mathrm{NP}}

\newcommand{\cG}{\mathcal{G}}
\newcommand{\cV}{\mathcal{V}}
\newcommand{\cE}{\mathcal{E}}

\newcommand{\juba}[1]{\textcolor{blue}{[Juba: #1]}}
\newcommand{\ar}[1]{\textcolor{red}{[Aaron: #1]}}
\newcommand{\comm}[1]{\textcolor{green}{[Comment: #1]}}
\newcommand{\sk}[1]{\textcolor{green}{[Sampath: #1]}}
\newcommand{\esh}[1]{\textcolor{cyan}{[Eshwar: #1]}}

\renewcommand{\juba}[1]{\iffalse #1 \fi}
\renewcommand{\ar}[1]{\iffalse #1 \fi}
\renewcommand{\comm}[1]{\iffalse #1 \fi}
\renewcommand{\sk}[1]{\iffalse #1 \fi}
\renewcommand{\esh}[1]{\iffalse #1 \fi}

\title{Pipeline Interventions}

\author{Eshwar Ram Arunachaleswaran\thanks{University of Pennsylvania. Email: eshwar@seas.upenn.edu.}
\and Sampath Kannan\thanks{University of Pennsylvania. Email: kannan@cis.upenn.edu.}
\and Aaron Roth\thanks{University of Pennsylvania. Email: aaroth@cis.upenn.edu.}
\and Juba Ziani\thanks{University of Pennsylvania. Email: jziani@seas.upenn.edu.}
}

\begin{document}

\maketitle

\begin{abstract}
We introduce the \emph{pipeline intervention} problem, defined by a layered directed acyclic graph and a set of stochastic matrices governing transitions between successive layers. The graph is a stylized model for how people from different populations are presented opportunities, eventually leading to some reward. In our model, individuals are born into an initial position (i.e. some node in the first layer of the graph) according to a fixed probability distribution, and then stochastically progress through the graph according to the transition matrices, until they reach a node in the final layer of the graph; each node in the final layer has a \emph{reward} associated with it. The pipeline intervention problem asks how to best make costly changes to the transition matrices governing people's stochastic transitions through the graph, subject to a budget constraint. We consider two objectives: social welfare maximization, and a fairness-motivated maximin objective that seeks to maximize the value to the population (starting node) with the \emph{least} expected value. We consider two variants of the maximin objective that turn out to be distinct, depending on whether we demand a deterministic solution or allow randomization. For each objective, we give an efficient approximation algorithm (an additive FPTAS) for constant width networks. We also tightly characterize the ``price of fairness'' in our setting: the ratio between the highest achievable social welfare and the social welfare consistent with a maximin optimal solution. Finally we show that for polynomial width networks, even approximating the maximin objective to any constant factor is NP hard, even for networks with constant depth. This shows that the restriction on the width in our positive results is essential.
\end{abstract}

\section{Introduction}

Inequality can be difficult to correct by the time it manifests itself in consequential domains. For example, faculty in computer science departments are disproportionately male~(\citet{gender}), and although the reasons for this are varied and complex, it seems difficult to correct \emph{only} by intervening in the process of faculty hiring (although the solution likely involves some intervention at this stage). The problem is that interventions at the final stage of a long pipeline may not be enough (or the best way) to address iniquities that compound starting from earlier stages in the pipeline such as graduate school, college, high school, enrichment programs, all the way back to birth circumstances. Because each stage of, for example, employment pipelines feeds into the next, interventions that are isolated to any one stage can have difficulty controlling effects on final outcomes --- and although in practice it is difficult to fully understand such a system, we would ideally like to design proposed interventions at a system-wide level, rather than myopically. 

Thus motivated, we study an optimization problem within a stylized (and highly simplified) model of such a pipeline. Our model is a layered directed acyclic graph. The vertices in the first layer represent a coarse partitioning  of possible birth circumstances into a small number of types --- each vertex representing one of these types. There is a probability vector over these vertices and individuals are ``born'' into some vertex with these probabilities. The graph represents a Markov process that determines how individuals progress through the pipeline. From every vertex there is a stochastic transition matrix specifying the probability that an individual will progress to each vertex in the next layer of the pipeline. We might imagine, for example, that the proportion of children that enroll in each of several elementary schools (the second layer of such a pipeline) varies according to the neighborhood that they are raised in (the first layer). The proportion of children that then go on to enroll in each of several high schools may then vary according to the elementary school they attend, and so on. Finally, vertices at the last layer of the pipeline are associated with payoffs. One may then calculate the expected payoff of an individual as a function of their initial position. These payoffs may vary widely depending on this position. 

We are concerned with the problem of how best to invest limited resources so as to \emph{modify} the transition matrices governing different layers of this pipeline to achieve some goal. In the main body of the paper, we focus on a stylized model where the costs of modifying transition matrices are linear, for simplicity of exposition; we extend our results to more complex and realistic cost functions in the Appendix. We consider two goals: the first is simply maximizing social welfare --- the expected payoff for an individual chosen according to the given probability vector for the first layer. Although this is a natural objective, it can easily lead to solutions that are ``unfair'' in the sense that they will prioritize investments that lead to improvements for majority populations over minority populations, simply because majority populations, by their sheer numbers, contribute more to social welfare. The second goal we study is therefore to maximize the \emph{minimum} expected payoff of individuals, where the minimum is taken over all of the initial positions, i.e., layer 1 vertices. This ``maximin'' objective is a standard fairness-motivated objective in allocation problems \citep[see, e.g.,][]{maxmin1,maxmin2,maxmin3}. In fact, we study two different variants of this objective, that can be distinguished by the timing with which one wants to evaluate fairness. The \emph{ex-ante} maximin objective asks for a \emph{distribution} over budget-feasible modifications of the transition matrices, that maximize the minimum expected payoff over all initial positions. The \emph{ex-post} maximin objective asks for a single (i.e. deterministic) budget-feasible modification to the transition matrices. Because the problem we study is non-convex, these two goals are distinct --- which is preferred depends on \emph{when} one wants to evaluate the fairness of a solution: before or after the randomization. 

\subsection{Overview of Our Results}
Briefly, our main contributions are the following:
\begin{enumerate}
    \item We define and formalize the \emph{pipeline intervention problem} with the social welfare, ex-ante maximin, and ex-post maximin objectives. We also prove a separation between the ex-post and ex-ante maximin solutions. 
    \item We give an additive fully polynomial-time approximation scheme (FPTAS) for both the social welfare and ex-post maximin objectives for networks of constant width (but arbitrarily long depth). \label{list:2}
    \item We give an efficient reduction from the ex-ante maximin objective problem to the ex-post maximin objective problem via equilibrium computation in two-player zero-sum games. Combined with our results from \ref{list:2}, this yields an additive FPTAS for the ex-ante maximin objective problem for constant width networks as well. 
    \item We define and prove tight bounds on the ``price of fairness'', which compares the optimal social welfare that can be achieved with a given budget to the social welfare of ex-post maximin optimal solutions.
    \item Finally, we show that the pipeline intervention problem is NP hard even to approximate in the general case when the width $w$ is not bounded --- and hence that our efficient approximation algorithms cannot be extended to the general case (or even the case of constant depth, polynomial width networks).
\end{enumerate}

\subsection{Related Work}
There is an enormous literature in ``algorithmic fairness'' that has emerged over the last several years, that we cannot exhaustively summarize here --- but see \citet{RC18} for a recent survey. Most of this literature is focused on the myopic effects of a single intervention, but what is more conceptually related to our paper is work focusing on the longer-term effects of algorithmic interventions.

\citet{DL18} and \citet{pipelines} study the effects of imposing fairness constraints on machine learning algorithms that might be composed together in various ways to reach an eventual outcome. They show that generally fairness constraints imposed on constituent algorithms in a pipeline or other composition do not guarantee that the same fairness constraints will hold on the entire mechanism as a whole. (They also study conditions under which fairness guarantees \emph{are} well behaved under composition).  Two recent papers~(\citet{delayed,delayed2}) study parametric models by which classification interventions in an earlier stage can have effects on the data distribution at later stages, and show that for many commonly studied fairness constraints, their effects can either be positive or negative in the long term, depending on the functional form of the relationship between classification decisions and changes in the agent type distribution.

There is also a substantial body of work studying game theoretic models for how interventions affect ``fairness'' goals. This work dates back to \citet{CL93,FV92} in the economics literature, who propose game theoretic models to rationalize how unequal outcomes might emerge despite two populations being symmetrically situated. More recently, in the computer science literature, several papers consider more complicated models that are similar in spirit to \citet{CL93,FV92}. \citet{HC18} propose a  two-stage  model of a labor market with a ``temporary'' (i.e. internship) and ``permanent'' stage, and study the equilibrium effects of imposing a  fairness constraint on the temporary stage. \citet{fat20} consider a model of the labor market with higher dimensional signals, and study equilibrium effects of subsidy interventions which can lessen the cost of exerting effort. \citet{downstream} study the effects of admissions policies on a two-stage model of education and employment, in which a downstream employer makes rational decisions. \citet{endogenous} study a model of criminal justice in which crime rates are responsive to the classifiers used to determine criminal guilt, and study which fairness constraints are consistent with the goal of minimizing crime.

\section{Model}

The \emph{pipeline intervention} problem is defined by a layered directed acyclic graph $G = (V,E)$, where $V$ is the set of vertices (or nodes), and $E$ is the set of edges. The vertices are partitioned into $k$ \emph{layers} $L_1,L_2,\cdots L_k$ , each consisting of $w$ vertices. We say that $w$ is the \emph{width} of the graph.  For every $t \in [k-1]$, there is a directed edge from every $u \in L_t$ to every $v \in L_{t+1}$; the graph contains no other edge. In turn, every path from layer $L_1$ to layer $L_k$ must go through exactly one vertex in each layer $L_2, \ldots,L_{k-1}$ in this order. Intuitively, such a layered graph represents a pipeline, in which individuals start at initial positions in layer $1$, and transition through the graph to final positions in layer $k$, stochastically according to transition matrices which we define next.
This layered model can be used to abstractly represent real-life pipelines; such a pipeline, that has received attention in previous work (e.g.~\citet{downstream}), is the education and job market one. Nodes in the initial layer represent a coarse partitioning of the population based on family income levels and educational background. The second layer could represent pre-K experience. For example, one could have 3 nodes in the second layer representing no pre-K, Headstart, and private pre-K. See for example~\citet{Headstart} for a general discussion of as well as pointers to recent studies on the efficacy of Headstart programs. At the next level or two, nodes can represent different qualities of K-12 schools, based on a coarse partitioning of their performance under one of several widely-available metrics, such as the ones provided by~\citet{usnews,niche}.

The layer after that could be a coarse partitioning where nodes represent, for example, no college, technical or vocational school, and 2 and 4-year colleges coarsely grouped together based on perceived quality according to one of several college rankings. A subsequent layer could encode the details of a student's performance in college, such as their major and GPA, again under a coarse bucketing. The last layer, with numerical rewards could represent different types of employment with rewards determined by starting salaries and prospects for advancement. 

In a more accurate model, we might perhaps condition the probability of transition from node $u$ in layer $i$ to node $v$ in layer $i+1$ on the entire path taken by an individual leading up to node $u$. However, for mathematical tractability, we make the simplifying assumption that the process is Markovian, and this transition probability from $u$ is independent of prior history.

Let $\cM$ be the set of left stochastic matrices in $\reals^{w \times w}$: i.e., $M \in \cM$ if and only if for all $j \in [w]$, $\sum_{i \in [w]} M(i,j) = 1$, and for all $(i,j) \in [w]^2$, $M(i,j) \geq 0$. Let $\cD \triangleq \left\{x \in [0,1]^w:~\sum_{k=1}^w x(k) = 1 \right\}$ be the set of probability distributions over $[w]$. An instance of the pipeline intervention problem is defined by three elements:
\begin{enumerate}
    \item A set of initial transition matrices $M_t^0 \in \cM$ between layers $L_t$ and $L_{t+1}$, for all $t \in [k-1]$, such that for all $u \in L_t$, $v \in L_{t+1}$, $M_t^0(v,u)$ denotes the probability of transitioning from node $u$ to node $v$. Note that we will multiply any input distribution to the right of any transition matrix we use in the paper.
    \item An input distribution $D_1$ over the vertices in layer $1$, where $D_1(u)$ denotes the fraction of the population that starts at $u$ in $L_1$ as their initial position. Without loss of generality we assume $D_1(u) > 0$ for all initial positions $u \in L_1$.
    \item Finally, a reward $R(v) \geq 0$ corresponding to each vertex $v \in L_k$ in the final layer. We let $R = \left(v\right)^\top_{v \in L_k}$ denote the vector of all rewards on layer $k$. We assume without loss of generality that the rewards on any two vertices in the final layer are distinct: for all $v,v' \in L_k$,  $R(v) \neq R(v')$. We can also assume without loss of generality (up to renaming) that $R(1) > \ldots > R(w)$.
\end{enumerate}

In our model, each vertex $u$ in the \emph{starting layer} $L_1$ represents the initial position of some population; abusing notation, we refer to this population also as $u$.  An individual in population $u$ transitions to a node in layer $L_2$, then a node in layer $L_3$, up until they reach a node $v$ in \emph{destination layer} $L_k$, and obtains a reward of $R(v)$, with probability given by the transition matrices $M_1^0$ to $M_{k-1}^0$. The expected reward of an individual from population $u$ is therefore given by $R^\top M_{k-1}\cdot \ldots\cdot M_1 e_u$, where $e_u$ represents the $w$-dimensional standard basis vector corresponding to index $u$. The aim of the pipeline intervention problem is to \emph{modify} the transition matrices between pairs of adjacent layers so as to improve these expected rewards in some way (we study several objectives) given a finite resource constraint.

We will take the point of view of a centralized designer, who can invest money into modifying the transition matrices between layers. We assume some edges can be modified, while some edges cannot; the edges that can be modified are called \emph{malleable}, and the edges that cannot be modified are called \emph{non-malleable}. We denote the set of malleable edges between layers $t$ and $t+1$ by $E^t_{mal}$   and the set of non-malleable edges by $\overline{E^t_{mal}}$ its complement. Further, we assume that modifying these transitions matrices comes at a cost, and that on a given layer $t$, the cost of transforming $M_t^0$ to some alternative $M_t \in \cM$ is given by:
\[
c(M_t,M_t^0) \triangleq \sum_{(i,j) \in [w]^2} \left\vert M_t(i,j) - M_t^0(i,j) \right\vert. 
\]
%We remark that although for simplicity we assume that costs are uniform in the sense that changing the probability of transitions is the same between any pair of vertices, our results naturally extend to non-uniform transition costs. 

\begin{remark}
A critique of such cost functions is that they may not be rich enough to model the cost of improving transitions and opportunities between different stages of, say, the education pipeline.

To address this, we note that while we focus on these simple cost functions in the main body of the paper for simplicity of exposition, our algorithmic results (of Sections~\ref{sec:SW_max},~\ref{sec:expost_maximin} and~\ref{sec:exante_maximin}) extend to more general and possibly more realistic cost functions --- so long as they are convex and increase at least linearly as the distance $\sum_{(i,j) \in [w]^2} \left\vert M_t(i,j) - M_t^0(i,j) \right\vert$ between modified transition matrix $M_t$ and initial transition matrix $M_t^0$ increases. We discuss this extension in more detail in Appendix~\ref{app:Lipschitz}.

This extension allows us to model more realistic situations such as those where the cost functions are not linear, but also those where different edges have different costs --- as motivated by the fact that real-life interventions often become more expensive the later they happen.
\end{remark}

The designer has a total budget of $B$, and can select target transition matrices  $(M_1,\ldots,M_{k-1})$ so long as the cost of modifying the initial transition matrices to his targets does not exceed his budget, and only malleable edges have been modified. That is, he must select target transition matrices subject to the constraint:
\[
\sum_{t=1}^{k-1} c(M_t,M_t^0) \leq B.
\]
We let 
\begin{align*}
&\mathcal{F}\left(B,M_1^0,\ldots,M_{k-1}^0\right) 
\\&= \left\{
\left(M_1,\ldots,M_{k-1}\right):~\sum_{t=1}^{k-1} c\left(M_t,M_t^0 \right) \leq B,~M_t \in \cM~\forall t,~M_t(i,j) = M_t^0(i,j)~\forall (i,j) \in \overline{E^{mal}_t}
\right\}
\end{align*}
be the set of feasible sets of transition matrices, given initial matrices $M_1^0,\ldots,M_{k-1}^0$ and budget $B$. We will consider several  objectives that we may wish to optimize. The first is simply to maximize the overall social welfare  (i.e. the expected reward of an individual chosen according to $D_1$), which is given by 
\[
W \left(M_1,\ldots,M_{k-1}\right) \triangleq R^\top M_{k-1} \ldots M_1 D_1.
\]
The second objective aims to compute a ``fair'' outcome in the sense that it evaluates a solution according to the expected payoff of the worst-off members of society (here interpreted as individuals starting at the pessimal initial position), rather than according to the average. This is the classic maximin objective. It turns out that there are two distinct variants of this problem, depending on whether one wishes to allow randomized solutions (i.e. distributions over matrices) or not. We will elaborate on this distinction in the next section, but in the deterministic variant we wish to optimize
\[
\min_{j \in [w]} R^\top M_{k-1} \ldots M_1 e_j,
\]
where $e_j \in \reals^w$ is the unit vector with $e_j(j) = 1$, and $e_j(i) = 0$ for all $i \neq j$. 

\begin{remark}\label{rem:simplifications}
%In defining the model above, we have made two simplifications that make our notation simpler, but are not necessary. These are:
%\begin{enumerate}
    We have assumed that each layer has \emph{exactly} $w$ vertices. In fact, all of our results generalize to the case in which each layer has $\leq w$ vertices.
    %\item We have assumed that \emph{all} edges are able to be modified. In fact, all of our algorithms generalize to the case in which only a subset of the edges between adjacent layers are ``malleable'' and the remaining are fixed (and probability mass can only be moved between malleable edges).
%\end{enumerate}
\end{remark}

\subsection{Optimization Problems of Interest}
In this paper, we will provide algorithms to solve the following three optimization problems. We note at the outset that these optimization problems are non-convex, due to the fact that our objective values are not convex for $k \geq 2$. Hence we should not expect efficient algorithms in the fully general setting; we will give efficient algorithms for networks of constant width $w$ (i.e. algorithms whose running time is polynomial in the depth of the network $k$), and show that outside of this class, the problem is NP hard even to approximate.

\paragraph{Social welfare maximization}
The first optimization problem we aim to solve is that of maximizing the social welfare of our network, under our budget constraint:
\begin{align}\label{SW_program}
   \begin{split}
                OPT_{SW}=
                \max_{M_1, \ldots, M_{k-1}}&~~R^\top M_{k-1} \ldots M_1 D_1^0
                \\\text{s.t.}&~~ \left(M_1,\ldots,M_{k-1}\right) \in \cF\left(B,M_1^0,\ldots,M_{k-1}^0\right)
        \end{split}
\end{align}

\paragraph{Ex-post maximin problem}
The second optimization problem aims to maximize the minimum expected reward that a population can obtain, where the minimum is taken over all initial positions: 
\begin{align}\label{maxmin_program}
\begin{split}
OPT_{MM} = 
\max_{M_1,\ldots,M_{k-1}}&~~ \min_{j \in [w]} R^\top M_{k-1} \ldots M_1 e_j
\\\text{s.t.}&~~ \left(M_1,\ldots,M_{k-1}\right) \in \cF\left(B,M_1^0,\ldots,M_{k-1}^0\right)
\end{split}
\end{align}

\paragraph{Ex-ante maximin problem}
The third optimization problem has the same objective as Program~\ref{maxmin_program}, but allows randomization over sets of transition matrices that satisfy the budget constraint. Note that the budget constraint must be satisfied \emph{ex-post}, for \emph{any realization} of the set of transition matrices. To define this optimization problem, we let $\Delta \cF\left(B,M_1^0,\ldots,M_{k-1}^0\right)$ the set of probability distributions with support $\mathcal{F}\left(B,M_1^0,\ldots,M_{k-1}^0\right)$. The optimization program is given by:
\begin{align}\label{exante_maxmin_program}
\begin{split}
OPT_{RMM} = 
\max_{\Delta M}&~~ \min_{j \in [w]} R^\top \E_{M \sim \Delta M} \left[M_{k-1} \ldots M_1\right] e_j
\\\text{s.t.}&~~\Delta M \in \Delta\cF\left(B,M_1^0,\ldots,M_{k-1}^0\right),
\end{split}
\end{align}
where the expectation is taken over the randomness of distribution $\Delta M$. Note that where Program~\ref{exante_maxmin_program} can be viewed as optimizing an \emph{ex-ante} notion of fairness, in which we are evaluated on the minimum expected value of individuals starting at any initial position, \emph{before the coins of $\Delta M$ are flipped.} In contrast, Program~\ref{maxmin_program} evaluates the minimum expected value of individuals starting at any initial position for an already established set of transition matrices. 

\begin{remark}
Programs~\eqref{SW_program},~\eqref{maxmin_program} and~\eqref{exante_maxmin_program} all have solutions, and as such the use of maxima instead of suprema is well defined. To see this, first note that the feasible sets are non-empty since $\left(M_1^0,\ldots,M_{k-1}^0\right) \in \cF\left(B,M_1^0,\ldots,M_{k-1}^0\right)$ for all $B \geq 0$. For Program~\eqref{SW_program}, the existence of a maximum is an immediate consequence of the fact that the objective function is continuous in $\left(M_1,\ldots,M_{k-1}\right)$ and $\cF$ and $\cM$ are compact sets. For Program~\eqref{maxmin_program}, note that no solution can have $R^\top M_{k-1} \ldots M_1 e_j \geq \Vert R \Vert_{\infty}$ for any $j$, as $M_{k-1} \ldots M_1 e_j$ is a probability distribution. Hence, we can rewrite the program as 
\begin{align*}
\begin{split}
\max_{v,M_1,\ldots,M_{k-1}}&~~v
\\\text{s.t.}&~~0 \leq v \leq \Vert R \Vert_{\infty},
\\&~~R^\top M_{k-1} \ldots M_1 e_j \geq v~\forall j \in [w],
\\&~~ \left(M_1,\ldots,M_{k-1}\right) \in \cF\left(B,M_1^0,\ldots,M_{k-1}^0\right).
\end{split}
\end{align*}
This is an optimization problem with a continuous objective function over a compact set, so it admits a solution. A similar argument follows for Program~\eqref{exante_maxmin_program}.
\end{remark}

\section{Algorithmic Preliminaries}

Our paper uses a dynamic programming approach for solving programs~\eqref{SW_program} and~\eqref{maxmin_program}. (Our solution to program~\eqref{exante_maxmin_program} is a game-theoretic reduction to our solution to program \eqref{maxmin_program}). Our algorithms will search over possible input distributions in $\cD$ starting from layer $L_t$ for all $t \in \{2,\ldots,k-2\}$, and over possible ways of splitting the total budget $B$ and allocating budget $B_t$ to the transition from layer $L_t$ to layer $L_{t+1}$, for all $t \in [k-1]$. To do so, we will need to discretize both the budget space $[0,B]$ and the probability space $\cD$.

\paragraph{Cost of Discretizing the Budget}
To discretize the budget space, we define $\cB(\varepsilon) = \left\{k\varepsilon,~\forall k \in \mathcal{N} \right\}$ to be the set of numbers on the real line that are multiples of $\varepsilon$. We consider the following discretized version of Programs~\ref{SW_program} and~\ref{maxmin_program} (We do not need to explicitly consider Program~\eqref{exante_maxmin_program}, since our solution for this one will be a reduction to our solution to Program~\eqref{maxmin_program}):
\begin{align}\label{dscrt_SW_program}
   \begin{split}
                OPT_{SW}^\varepsilon = \max_{M_1, \ldots, M_{k-1}}&~~R^\top M_{k-1} \ldots M_1 D_1
                \\\text{s.t.}&~~ c(M_t,M_t^0) \leq B_t~\forall t \in [k-1]
                \\&~B_t \in \cB(\varepsilon)~\forall t \in [k-1],~\sum_{t=1}^{k-1} B_t \leq B
                \\&~M_t(i,j) = M_t^0(i,j)~\forall (i,j) \in \overline{E^{mal}_t},~M_t \in \cM~\forall t
        \end{split}
\end{align}
and 
\begin{align}\label{dscrt_maxmin_program}
   \begin{split}
                OPT_{MM}^\varepsilon = \max_{M_1, \ldots, M_{k-1}}&~~\min_{j \in [w]} R^\top M_{k-1} \ldots M_1 e_j
                \\\text{s.t.}&~~ c(M_t,M_t^0) \leq B_t~\forall t \in [k-1]
                \\&~B_t \in \cB(\varepsilon)~\forall t \in [k-1],~\sum_{t=1}^{k-1} B_t \leq B
                \\&~M_t(i,j) = M_t^0(i,j)~\forall (i,j) \in \overline{E^{mal}_t},~M_t \in \cM~\forall t.
        \end{split}
\end{align}
We show that this discretization does not affect the optimal value of our problems by much:
\begin{claim}\label{clm: budget_loss}
There exists a feasible solution $\left(M_1^\varepsilon,\ldots,M_{k-1}^\varepsilon\right)$ to Program~\eqref{dscrt_SW_program} (resp. Program~\eqref{dscrt_maxmin_program}) with objective value at least $OPT_{SW} - (k-1) \varepsilon \left\Vert R \right\Vert_{\infty}$ 
(resp. $OPT_{MM}- (k-1) \varepsilon \left\Vert R \right\Vert_{\infty}$).
\end{claim}

We provide a brief proof sketch below, and defer the full proof to Appendix~\ref{app:budget_loss}.
\begin{proof}[Proof Sketch]
We prove this result by constructing transition matrices $M_t^\varepsilon$ that use roughly $\varepsilon$ budget less than $M_t^*$. We show that we can do so so as to only lose welfare of the order of $\epsilon$ in each of the $k-1$ layer transitions we consider, and that this loss composes additively.
\end{proof}

\begin{definition}
Let $K \subseteq \reals^{w}$. We call a subset $S$ of $K$ an $\varepsilon$-net for $K$ with respect to the $\ell_1$-norm if and only if for every $D \in K$, there exists $D' \in S$ such that 
\[
\Vert D - D' \Vert_1 \leq \varepsilon.
\]
\end{definition}

\begin{claim}[$\varepsilon$-nets in $\ell_1$-distance for $\cD$]\label{clm:eps_net}
Take $\varepsilon > 0$. There exists an $\varepsilon$-net $\cD(\varepsilon)$ of $\cD$ with respect to the $\ell_1$-norm that has size $\left(\frac{1}{\varepsilon}\right)^w$.
\end{claim}
%We let $\cD(\varepsilon)$ be such a net in the rest of the paper.
This is a standard proof, included in Appendix \ref{app:eps_net} for completeness.

\section{Social Welfare Maximization}\label{sec:SW_max}

We want to solve the following optimization problem: 
\begin{align}
   \begin{split}
                \max_{M_1, \ldots, M_{k-1}}&~~R^\top M_{k-1} \ldots M_1 D_1
                \\\text{s.t.}&~~ \sum_{t=1}^{k-1} c\left(M_t, M_t^0\right) \leq B,
                \\&~~M_t \in \cM~\forall t \in [k-1],
    \end{split}
\end{align}

\subsection{A Dynamic Programming Algorithm for Social Welfare Maximization}
\label{subsection:dpsection}
In this section, we describe a dynamic programming algorithm for approximately solving the problem above on long skinny networks. The algorithm will run in polynomial time when the width $w$ of the network is small; its running time is polynomial in the depth $k$ of the network, but exponential in the width $w$. The formal description is given in Algorithm~\ref{alg: max_SW}. Our algorithm works backwards, starting from the final transition matrix from layer $L_{k-1}$ to $L_{k}$. It builds up the solutions to sub-problems parameterized by 
three parameters --- a layer $t$, a starting distribution over the vertices in layer $t$, and a budget $B_{\geq t}$ that can be used at layers $\geq t$. For each sub-problem, it computes an approximately welfare-optimal solution. Once all of these sub-problems have been solved, the optimal solution to the original problem can be read off from the ``sub-problem'' in which $t = 1$, the starting distribution is the distribution on initial positions, and $B_{\geq 1} = B$. Here is the informal description of the algorithm: 
\begin{enumerate}

    \item For $t$ going backwards from $k-1$ to $1$, the algorithm does the following exploration over budget splits and  probability distributions $D_t, D_{t+1} \in \cD(\varepsilon)$ (an $\varepsilon$-net for the $w$-dimensional simplex in $\ell_1$ norm) on $L_t$:
    \begin{enumerate}
    \item The algorithm explores all discretized splits of a budget $B_{\geq t}$ to be used for layers $t$ to $k-1$ into a budget $B_t$ to expend on layer $t$ and a budget $B_{\geq t + 1}$ to expend on the remaining layers $t+1$ to $k-1$, as well as all choices of target output probability distribution $D_{t+1} \in \cD(\varepsilon)$ on layer $L_{t+1}$ and the starting probability distribution $D_{t} \in \cD(\varepsilon)$. Informally, we can think of these ``target'' and ``initial'' probability distributions as guesses for what the distribution on vertices in layer $t+1$ and layer $t$ look like in the optimal solution. Recall that for each $D_{t+1}$ and $B_{\geq t+1}$, our algorithm has already computed a near-optimal solution for a smaller sub-problem, which we will utilize in the next step. 
    \item\label{step2} The algorithm then finds a transition matrix  from $L_t$ to $L_k$ that maximizes welfare when the starting distribution on layer $t$ is $D_t$ and the remaining transition matrices are fixed as in the solution to the corresponding sub-problem. Although the overall welfare-maximization problem is non-convex, this sub-problem can be solved as a linear program (Program~\ref{SW_DP}) because all transition matrices except for one have been fixed as the solution to our sub-problem. 
    \item Finally, the algorithm picks and stores the recovered transition matrices from layer $L_t$ to $L_k$ that yield the highest reward, among all the transition matrices recovered from step~\ref{step2}.  
    \end{enumerate}
\end{enumerate}

We remark that while (for notational simplicity) our algorithm is written as if all layers have size exactly $w$, it can easily be extended to the case in which all layers have size \emph{at most} $w$.

\begin{algorithm}[ht!]
\SetAlgoLined
\KwIn{Input distribution $D_1$, reward vector $R$, initial transition matrices $M_1^0,\ldots,M_{k-1}^0$, budget $B$, discretization parameter $\varepsilon$.}
\KwOut{Transition $M(B^\varepsilon,D_1^0)$ from $L_1$ to $L_k$.}

\textbf{Initialization:} Let $B_{\geq k} = 0$, $M(B_{\geq k}, D_{k}) = I$, $B^\varepsilon = \max \{x \in \cB(\varepsilon):~x \leq B\}$.\\
    \For{layer $t = k-1, \ldots, 1$}{
        \For{all distributions $D_t \in \cD(\varepsilon)$ if $t \neq 1$ ($D_t = D_1^0$ if $t = 1$) and budgets $B_{\geq t} \in \cB(\varepsilon)$ with $B_{\geq t} \leq B$}{
            \For{all distributions $D_{t+1} \in \cD(\varepsilon)$ and budgets $B_{\geq t+1} \leq B_{\geq t}$ such that $B_{\geq t+1} \in \cB(\varepsilon)$}{
            Solve linear program
            
            \begin{align}\label{SW_DP}
                \begin{split}
                M_t(B_{\geq t},B_{\geq t+1},D_t,D_{t+1}) 
                = \argmax_{M_t}&~~R^\top M(B_{\geq t+1}, D_{t+1}) M_t D_t
                \\\text{s.t.}&~~ c\left( M_t, M_t^0 \right) \leq B_{\geq t} - B_{\geq t+1} ,%~\forall t \in [k-1],
                \\&~~M_t(i,j) = M_t^0(i,j)~\forall (i,j) \in \overline{E_t^{mal}}
                \\&~~M_t \in \cM%~\forall t \in [k-1].
                \end{split}
                \end{align}
            }
            Pick $B_{\geq t+1}, D_{t+1}$ leading to the highest objective value in Program~\ref{SW_DP}, and set $M(B_{\geq t},D_{t}) = M(B_{\geq t+1}, D_{t+1}) M_t(B_{\geq t},B_{\geq t+1},D_t,D_{t+1})$.\\ 
        }
    }
\textbf{Return} $M(B^\varepsilon,D_1)$.
\caption{Dynamic Program for (Approximate) Social Welfare Maximization.}\label{alg: max_SW}
\end{algorithm}

We briefly note why Program~\eqref{SW_DP} is a linear program. The objective is linear because only the matrix $M_t$ represents variables. Thus we simply need to verify that the constraint on the cost is linear.

\begin{definition}
We say that a transition matrix $M_t \in \cM$ is feasible with respect to a budget split $B_{\geq t}, B_{\geq t+1}$ if and only if 
\[
c\left(M_t,M_t^0 \right) \leq B_{\geq t+1} - B_{\geq t}. 
\]
and $M_t(i,j) = M_t^0(i,j)$ for every non-malleable edge (i,j).
\end{definition}
Note that saying that $M_t$ feasible with respect to $B_{\geq t}, B_{\geq t+1}$ is equivalent to saying that $M_t$ is a feasible solution to Program~\eqref{SW_DP} with parameters $B_{\geq t},B_{\geq t+1},D_t,D_{t+1}$ for any $D_t, D_{t+1} \in \cD(\varepsilon)$. %\ar{Should this be $\cD(\varepsilon)$?}\juba{yep} 
The constraint $c\left(M_t,M_t^0 \right) \leq B_{\geq t+1} - B_{\geq t}$ can be equivalently replaced by $2w^2 +1$ linear constraints. To do so, we introduce $w^2$ variables - $a_1,a_2,\cdots a_{w^2}$. The constraint can then be rewritten in the form $\sum_{i=1}^{w^2} |f_i| \le B_{\geq t+1} - B_{\geq t}$, where each $f_i$ is a linear combination of the variables. We can thus express the budget constraint of Program~\ref{SW_DP} by the following set of linear constraints:
\begin{enumerate}
    \item $f_i \le a_i~~\forall i \in [w^2]$
    \item$-f_i \le a_i~~\forall i \in [w^2]$
    \item $\sum_{i=1}^{w^2} a_i \le B_{\geq t+1} - B_{\geq t}$.
\end{enumerate} 
%Observe that any feasible solution under the new constraints satisfies the old constraints as well. Thus, since we have only added extra constraints, the new program cannot have a better optimum than the original program. Additionally, observe that setting $a_i = |f_i|$ for all $i \in [w^2]$ reduces the new program to the original program, implying that the optimum of the new program is at least as good as that of the original program. Consequently, both programs have the same optimum.\juba{I think this is a pretty standard trick/argument so we don't need to spell it out maybe? Aaron, Sampath, make an executive decision (:}
Thus, Program~\ref{SW_DP} can be written as a linear program with the number of constraints and variables being polynomial in $w$.

% \subsection{Algorithm for fat but short networks}
% \juba{We don't have a good solution as of now. The $w^{1/\varepsilon}$ size net idea we had does not work (i.e., we seem to need a net of size exponential in $w$ if we want it to be a net for the $l1$-norm).}

\subsection{Running Time and Social Welfare Guarantees}

We provide the running time and social welfare guarantees of Algorithm~\ref{thm: SW_approximation} below. 

\begin{theorem}\label{thm: SW_approximation}
Algorithm~\ref{alg: max_SW} instantiated with discretization parameter $\varepsilon$ yields a solution achieving social welfare at least $OPT - 3(k-1) \varepsilon \Vert  R \Vert_{\infty}$, and has running time $O\left(k \frac{B}{\varepsilon} \left(\frac{1}{\varepsilon}\right)^{w^2} f(w) \right)$, where $f(w)$ is any upper-bound on the running time for solving linear Program~\ref{SW_DP}, which is always polynomial in $w$.
\end{theorem}
% \ar{The next thing isn't a corollary unless the theorem refers to running time}\juba{Better? Also added the running time in the proof of the theorem.}

This immediately yields the following corollary:
\begin{corollary}\label{cor: SW_algo_guarantee}
Algorithm~\ref{alg: max_SW} with discretization parameter $\varepsilon' = \frac{\varepsilon}{3(k-1)}$ yields social welfare at least $OPT - \varepsilon \Vert R \Vert_{\infty}$, and has running time $O\left(k^2 \frac{B}{\varepsilon} \left(\frac{k}{\varepsilon}\right)^{w^2} f(w) \right)$, where $f(w)$ is any upper-bound on the running time for solving linear Program~\ref{SW_DP}, which is always polynomial in $w$.
\end{corollary}

We observe that this running time is polynomial in $k$ (the depth of the network) and $1/\varepsilon$ (the inverse additive error tolerance), but exponential in $w$ (the width of the network). Hence our algorithm runs in polynomial time for the class of constant width networks. 

\begin{remark}
We note that our additive near-optimality guarantee can be translated into a multiplicative guarantee. In the case where \emph{all edges are malleable}, this follows from noting that given budget $B$, $OPT \geq \frac{B}{2w} \Vert R \Vert_{\infty}$: this can be reached by investing the totality of the budget into transitioning every node in the second-to-last layer to the highest reward node in the last layer, with probability $\frac{B}{2w}$ for each such node. Taking $\varepsilon = \delta \cdot \frac{B}{6(k-1) w}$ for some constant $\delta < 1$ gives a multiplicative approximation to the optimal social welfare with approximation factor $1-\delta$. 

For the case in which non-malleable edges are allowed, a lower bound on $OPT$ is given by $OPT \geq W_0$. Taking $\varepsilon = \delta \cdot \frac{W_0}{3(k-1) \Vert R \Vert_{\infty}}$ yields a multiplicative $1-\delta$ approximation still.
\end{remark}

\paragraph{Proof of Theorem~\ref{thm: SW_approximation}} The proof of Theorem~\ref{thm: SW_approximation} relies on the following lemma, and its corollary:
\begin{lemma}
Let $M \in \reals^{w \times w}$ be a left stochastic matrix, and let $D, D' \in \cD$ be probability distributions.
\[
\Vert M D - M D' \Vert_1 \leq \Vert D - D' \Vert_1.
\]
\end{lemma}

\begin{proof}
Note that 
\begin{align*}
\Vert M (D-D') \Vert_1
= \sum_{i=1}^w \left\vert \left(M(D-D')\right)(i) \right\vert
&= \sum_{i=1}^w \left\vert \sum_{j=1}^w M(i,j) (D(j) - D'(j)) \right\vert
\\&\leq \sum_{i=1}^w \sum_{j=1}^w \left\vert M(i,j) (D(j) - D'(j)) \right\vert
\\&= \sum_{j=1}^w \left\vert D(j) - D'(j) \right\vert \sum_{i=1}^w \left\vert M(i,j) \right\vert
\\&= \sum_{j=1}^w \left\vert D(j) - D'(j) \right\vert
\\&= \Vert D - D' \Vert_1,
\end{align*}
where the inequality follows from the triangle inequality, and the second-to-last equality from the fact that
\[
\sum_{i=1}^w \left\vert M(i,j) \right\vert = \sum_{i=1}^w  M(i,j) = 1~~\forall j \in [w]
\]
as $M$ is a left stochastic matrix.
\end{proof}

\begin{corollary}\label{cor:approx_loss}
Let $R \in \reals^w$ be a real vector and $D,D' \in \cD$ be probability distributions such that $\Vert D - D' \Vert_1 \leq \varepsilon$, and $M \in \reals^{w \times w}$ a left stochastic matrix. Then
\[
R^\top M D \geq R^\top M D' - \Vert R \Vert_{\infty}  \cdot \varepsilon.
\]
\end{corollary}

\begin{proof}[Proof of Corollary~\ref{cor:approx_loss}]
$
\Vert R^\top M (D'-D)\Vert_1 
\leq \Vert R \Vert_{\infty} \Vert M (D'-D)\Vert_1 
\leq \Vert R \Vert_{\infty} \Vert D'-D\Vert_1 
\leq \Vert R \Vert_{\infty} \cdot  \varepsilon
$
, where the first step follows from Holder's inequality.
\end{proof}

%\juba{Changed proof to incorporate budget split. Please check.}
We are now ready to prove Theorem~\ref{thm: SW_approximation}:

%\juba{if not enough space, I guess we can make this a proof sketch/intuition, and put everything in Appendix.}

\begin{proof}[Proof of Theorem~\ref{thm: SW_approximation}]
Let us denote by $B_1^\varepsilon,\ldots, B_{k-1}^\varepsilon$  a split of the budget for the discretized problem with $B^\varepsilon_{\geq t} = B_t^\varepsilon + \ldots + B_{k-1}^\varepsilon$. Let $M_1^\varepsilon, \ldots, M_{k-1}^\varepsilon$ a set of transition matrices achieving welfare $R^\top M_{k-1}^\varepsilon \ldots M_1^\varepsilon D_1^0 \geq OPT^\varepsilon \triangleq OPT - (k-1) \varepsilon \left\Vert R \right\Vert_{\infty}$ that is feasible with respect to budget split $B_1^\varepsilon,\ldots, B_{k-1}^\varepsilon$. Note that such a budget split and matrices exist by Claim~\ref{clm: budget_loss}. Let $D_t^\varepsilon$ the probability distribution on layer $t$ defined by these transition matrices, i.e.
\[
D_t^\varepsilon = M^\varepsilon_{t-1} \ldots M^\varepsilon_1 D_1^0.
\]
To prove the result, we will show by induction that for all $B_{\geq t} \geq B^\varepsilon_{\geq t}$, and for $D_t \in \cD(\varepsilon)$ such that $\Vert D_t - D_t^\varepsilon \Vert_1 \leq \varepsilon$,
\[
R^\top M(B_{\geq t},D_t) D_t \geq OPT^\varepsilon - 2(k - t) \varepsilon \Vert  R \Vert_{\infty}.
\]
This will directly imply that as $B^\varepsilon$ is one of the possible values of $B_{\geq 1}$,
\[
R^\top M(B^\varepsilon,D_1) D_1 \geq OPT^\varepsilon - 2(k-1) \varepsilon \left\Vert R \right\Vert_{\infty}.
\]
Combined with Claim~\ref{clm: budget_loss} that states $OPT_\varepsilon \geq OPT - (k-1) \varepsilon \left\Vert R \right\Vert_{\infty}$, we will obtain the result. 

Let us now provide our inductive proof. First, consider the transition from layer $L_{k-1}$ to layer $L_k$. Note that 
\[
OPT^\varepsilon \le  R^\top M_{k-1}^\varepsilon\ldots M_1^\varepsilon D_1^0 = R^\top M_{k-1}^\varepsilon D_{k-1}^\varepsilon.
\]
Let $D_{k-1} \in \cD(\varepsilon)$ be such that $\Vert D_{k-1} - D^\varepsilon_{k-1} \Vert \leq \varepsilon$. Note then that by Corollary~\ref{cor:approx_loss}, 
\[
R^\top M_{k-1}^\varepsilon D_{k-1} \geq R^\top M_{k-1}^\varepsilon D^\varepsilon_{k-1} - \varepsilon \Vert  R \Vert_{\infty}.
\]
Further, $M_{k-1}^\varepsilon$ is feasible for Program~\eqref{SW_DP} with respect to $B_{\geq k-1}, B_{\geq k} = 0$, given $B_{\geq k-1} \geq B_{\geq k-1}^\varepsilon$. As such, for $B_{\geq k-1} \geq B_{\geq k-1}^\varepsilon$, we have that
\[
R^\top M(B_{\geq k-1},D_{k-1}) D_{k-1}  \geq R^\top M_{k-1}^\varepsilon D_{k-1},
\]
and in turn
\[
R^\top M(B_{\geq k-1},D_{k-1}) D_{k-1} 
\geq OPT^\varepsilon - \varepsilon \Vert  R \Vert_{\infty}.
\]

Now, suppose the induction hypothesis holds at layer $t+1$. I.e., for all $B_{\geq t+1} \geq B^\varepsilon_{\geq t+1}$, for $D_{t+1} \in \cD(\varepsilon)$ such that $\Vert D_{t+1} - D_{t+1}^\varepsilon \Vert_1 \leq \varepsilon$,
\[
R^\top M(B_{\geq t+1},D_{t+1}) D_{t+1} \geq OPT^\varepsilon - 2(k - t - 1) \varepsilon \Vert  R \Vert_{\infty}.
\]
For any $B_{\geq t} \geq B_{\geq t}^\varepsilon$, note that one can set $B_{\geq t+1} = B_{\geq t+1}^\varepsilon$ and $B_t \geq B_t^\varepsilon$; hence, $M_t^\varepsilon$ is feasible for Program~\eqref{SW_DP} with respect to $B_t \geq B_t^\varepsilon,B^\varepsilon_{\geq t+1}$. Since $\Vert D_t -D_t^\varepsilon \Vert_1 \leq \varepsilon$ and $\Vert D_{t+1} - M_t^\varepsilon D_t^\varepsilon \Vert_1 \leq \varepsilon$, we have that by Corollary~\ref{cor:approx_loss},
\[
R^\top M(B^\varepsilon_{\geq t+1}, D_{t+1}) M^\varepsilon_t D_t 
\geq R^\top M(B^\varepsilon_{\geq t+1}, D_{t+1}) M^\varepsilon_t D_t^\varepsilon - \varepsilon \Vert  R \Vert_{\infty}
\geq R^\top M(B^\varepsilon_{\geq t+1}, D_{t+1}) D_{t+1} - 2 \varepsilon \Vert  R \Vert_{\infty}.
\]
Using the induction hypothesis, we obtain that
\[
R^\top M(B^\varepsilon_{\geq t+1}, D_{t+1}) M^\varepsilon_t D_t \geq OPT^\varepsilon - 2(k - t) \varepsilon \Vert  R \Vert_{\infty}.
\]
In particular,
\[
R^\top M(B_{\geq t}, D_{t}) D_t \geq OPT^\varepsilon - 2(k - t) \varepsilon \Vert  R \Vert_{\infty},
\]
which concludes the proof of the social welfare guarantee. %\juba{running time added here} 
For the running time, we note that at each time step $t$, we solve one instance of Program~\ref{SW_DP} for each of the (at most) $\frac{B}{\varepsilon}$ possible budget splits of $B_{\geq t}$ and for each of the $\left(\frac{1}{\varepsilon}\right)^w$ (by Claim~\ref{clm:eps_net}) probability distributions in $\cD(\varepsilon)$ in layer $L_t$ and layer $L_{t+1}$; i.e., for each $t$, the algorithm solves $O \left(\frac{B}{\varepsilon} \left(\frac{1}{\varepsilon}\right)^{w^2}\right)$ optimization programs. Then, the algorithm finds the solution of all of these programs with the best objective value, which can be done in time linear in the number of such solutions, i.e. $O \left(\frac{B}{\varepsilon}\left(\frac{1}{\varepsilon}\right)^{w^2}\right)$. This is repeated for $k-1$ values of $t$.
\end{proof}

\section{(Ex-post) Maximin Value Maximization}\label{sec:expost_maximin}

Although social welfare maximization is a natural objective, it is well-known that it can be ``unfair'' in the sense that it explicitly prioritizes the welfare of larger populations (here represented as initial positions that have larger probability mass) over smaller populations. We can alternately evaluate a solution according to the welfare of the \emph{least-well-off} population (here represented by the initial position with the smallest expected value) and ask to optimize \emph{that} objective. We show how to optimize this objective in this section, when one demands a deterministic solution.

\subsection{A Dynamic Programming Algorithm for Computing an Ex-post Maximin Allocation}

\paragraph{Algorithm and proof:} 

In this subsection, we adapt the dynamic programming approach in Section \ref{subsection:dpsection} to give an approximation algorithm for the problem of maximizing the minimum expected reward over all initial positions. Recall that $\cD$, the probability simplex, denotes the set of all possible probability distributions on a layer. Intuitively, our algorithm for maximizing social welfare kept track of a single probability distribution in each subproblem: the overall probability of arriving at each vertex in the layer \emph{over both the randomness of an individual's initial position, and the randomness of the transition matrix}. In order to optimize the \emph{minimum} expected value over all initial positions, we will need to keep track of more state.  At every layer $L_t$, we will keep track of the probability of reaching each vertex in that layer \emph{from each initial position in the starting layer}. So,  we will now keep track of collections of $w$ probability distributions in $\cD^w$, one for each starting position. We call the elements of $\cD^w$ \emph{population-wise distributions}. 

We introduce a discretization $\cA(\varepsilon)$ of $\cD^w$, as follows: %which is a discretization of the set of all possible probability distributions of all agents on a given layer (i.e., the set $\cD^w$).
$\cA(\varepsilon) \triangleq (\cD(\varepsilon))^w$, where $\cD(\varepsilon)$ denotes a $\varepsilon$-net of $\cD$ (of size $\left(\frac{1}{\varepsilon}\right)^w$). Given a population-wise probability distribution $A_t \in \cA(\varepsilon)$ at layer $t$, we write $A_t^j$ for the probability distribution corresponding to population $j$. The algorithm works as follows, just as before, running backwards from the final layer to the first layer:
\begin{enumerate}
    \item For $t$ going backwards from $k-1$ to $1$, the algorithm does the following, for every  population-wise distribution $A_t \in \cA(\varepsilon)$ on $L_t$:
    \begin{enumerate}
    \item The algorithm explores all splits of the budget $B_{\geq t}$ for layers $t$ to $k$ into a budget $B_t$ for the transition from $L_t$ to $L_{t+1}$ and a budget $B_{\geq t + 1}$ for $L_{t+1}$ to $L_k$, as well as all choices of output population-wise probability distributions $A_{t+1} \in \cA(\varepsilon)$ on layer $L_{t+1}$.
    \item\label{stepp2} The algorithm then finds a near-optimal transition matrix from $L_t$ to $L_k$ for every budget decomposition, by using the previously computed  near-optimal solution for layers $L_{t+1}$ to $L_k$, and solving Program~\ref{MMW_DP}. The program maximizes the \emph{minimum reward obtained from any initial position}, assuming the population-wise distribution of individuals at layer $L_t$ is given by $\cA_t$. 
    \item Finally, the algorithm picks and stores the best recovered transition matrices from layer $L_t$ to $L_k$ that yield the highest reward, among all the transition matrices recovered from step~\ref{stepp2}.  
    \end{enumerate}
\end{enumerate}
The input population-wise probability distribution $A_1 \in \cD^w$ on the first layer is defined in the following manner, $A_1^j  := e_j$ (the $j$-th basis vector in the usual orthonormal basis of $\mathbb{R}^w$) for all $j \in [w]$.

\begin{algorithm}[ht!]
\KwIn{Reward vector $R$ on layer $L_k$, initial transition matrices $M_1^0,\ldots,M_{k-1}^0$, budget $B$, discretization parameter $\varepsilon$.}
\KwOut{Transition $M(B^\varepsilon,A_1)$ from $L_1$ to $L_k$.}
\textbf{Initialization:} Let $B_{\geq k} = 0$, $M(B_{\geq k}, A_k) = I_{w \times w}$, $B^\varepsilon = \max \{x \in \cB(\varepsilon):~x \leq B\}$.\\
    \For{layer $t = k-1, \ldots, 1$}{
        \For{all distributions $A_t \in \cA(\varepsilon)$ if $t \neq 1$ (resp. for $A_t = A_1$ if $t = 1$), budget $B_{\geq t} \in \cB(\varepsilon)$ with $B_{\geq t} \leq B$}{
            \For{all distributions $A_{t+1} \in \cA(\varepsilon)$, budget $B_{\geq t+1} \leq B_{\geq t}$ such that $B_{\geq t+1} \in \cB(\varepsilon)$}{
            
            Solve 
                \begin{align}\label{MMW_DP}
                \begin{split}
                M_t(B_{\geq t},B_{\geq t+1},A_t,A_{t+1}) 
                = \argmax_{M_t}&~~\min_{j \in [w]} R^\top M(B_{\geq t+1}, A_{t+1}) M_t A^j_t
                \\\text{s.t.}&~~ c\left( M_t, M_t^0 \right) \leq B_{\geq t} - B_{\geq t+1}~\forall t \in [k-1], 
                \\&~~M_t(i,j) = M_t^0(i,j)~\forall (i,j) \in \overline{E^{mal}_t}
                \\&~~M_t \in \cM.
                \end{split}
                \end{align}
            }
            Pick $B_{\geq t+1},A_{t+1}$ with the best objective value in Program~\ref{MMW_DP}, and let 
           \[
           M(B_{\geq t},A_{t}) \triangleq M(B_{\geq t+1}, A_{t+1}) M_t(B_{\geq t},B_{\geq t+1},A_t,A_{t+1}).
           \]
        }
    }
    \Return $M(B^\varepsilon,A_1)$
    
\caption{Dynamic Program for (Approximate) Maximin Value.}\label{alg: max_MM}
\end{algorithm}

Note that Program~\eqref{MMW_DP} can be written as a linear program of size polynomial in $w$, using the same method that was employed to write Program~\eqref{SW_DP} as a linear program.

\subsection{Running Time and Ex-Post Maximin Value Guarantees}

Remember that we let $OPT_{MM}$ denote the maximin value value of the given network. The running time and accuracy guarantees of Algorithm~\ref{alg: max_MM} are provided below:

\begin{theorem}\label{thm:MMW_approximation}
Algorithm~\ref{alg: max_MM} with discretization parameter $\varepsilon$ yields maximin value at least $OPT_{MM} - 3(k-1) \varepsilon \Vert  R \Vert_{\infty}$, and has running time $O\left(k \frac{B}{\varepsilon} \left(\frac{1}{\varepsilon}\right)^{w^4} g(w) \right)$, where $g(w)$ is any upper-bound on the running time for solving linear Program~\ref{MMW_DP}, which is always polynomial in $w$. 
\end{theorem}

This immediately induces the following corollary:
\begin{corollary}\label{cor: SW_algo_guarantee_MM}
Algorithm~\ref{alg: max_MM} with discretization parameter $\varepsilon' = \frac{\varepsilon}{3(k-1)}$ yields maximin value at least $OPT_{MM} - \varepsilon \Vert R \Vert_{\infty}$, and has running time $O\left(k^2 \frac{B}{\varepsilon} \left(\frac{k}{\varepsilon}\right)^{w^4} g(w)\right)$, where $g(w)$ is a polynomial upper-bound on the running time of linear Program~\ref{MMW_DP}.
\end{corollary}

The proof of Theorem~\ref{thm:MMW_approximation} is almost identical to that of Theorem~\ref{thm: SW_approximation}. We provide the full proof in Appendix~\ref{app:maximin_program}.

\section{(Ex-ante) Maximin Value Maximization}\label{sec:exante_maximin}
In this section, we consider the problem of optimizing the \emph{ex-ante} minimum expected value over all initial positions: in other words, we allow ourselves to find a \emph{distribution} over solutions, and take expectations over the randomness of this distribution, solving:
\begin{align}
\begin{split}
\max_{\Delta M}&~~ \min_{j \in [w]} R^\top \E_{M \sim \Delta M} \left[M_{k-1} \ldots M_1\right] e_j
\\\text{s.t.}&~~\Delta M \in \Delta\cF\left(B,M_1^0,\ldots,M_{k-1}^0\right)
\end{split}
\end{align}

We show in Appendix~\ref{app:separation} that this can yield strictly higher utility than optimizing the ex-post minimum value. We then give an algorithm for solving the ex-ante problem by exhibiting a game theoretic reduction to the ex-post problem. 

\subsection{Solving the Ex-ante Maximization Problem Using Algorithm~\ref{alg: max_SW}}

Because Program~\ref{exante_maxmin_program} is a $\max\min$ problem over a polytope, we can view it as a zero-sum game, and the solution that we want corresponds to a maxmin equilibrium strategy of this game. As first shown by \cite{FS96}, it is possible to compute an approximate equilibrium of a zero-sum game if we can implement a \emph{no-regret} learning algorithm for one of the players, and an approximate best-response algorithm for the other player: if we simply simulate repeated play of the game between a no-regret player and a best-response player, then the empirical average of player actions in this simulation converges to the Nash equilibrium of the game. 

This forms the basis of our algorithm. One player plays the ``multiplicative weights'' algorithm over the initial positions in layer 1 of the graph. This induces at every round a distribution over initial positions. The best response problem, which must be solved by the other player, corresponds to solving a welfare-maximization problem given the distribution over initial positions represented by the multiplicative weights distribution. Fortunately, this is exactly the problem that we have already given a dynamic programming solution for. The solution in the end corresponds to the uniform distribution over the solutions computed by the best-response player over the course of the dynamics. The algorithm is formally described below:
\begin{algorithm}
\KwIn{Time horizon $T$, reward vector $R$ on layer $L_k$, initial transition matrices $M_1^0,\ldots,M_{k-1}^0$, budget $B$, discretization parameter $\varepsilon$.}
\KwOut{$M^1,\ldots,M^T \in \cF\left(B,M_1^0,\ldots,M_{k-1}^0\right)$.}
\textbf{Initialization:} The no-regret player picks $D^1 = \left(\frac{1}{w},\ldots,\frac{1}{w}\right) \in \cD$, the uniform distribution over $[w]$.\\
\For{$t = 1, \ldots, T$}{
The no-regret player plays distribution $D^t \in \cD$.\\
The best-response player chooses $M^t \in \cF\left(B,M_1^0,\ldots,M_{k-1}^0\right)$ such that 
    \[
    R^\top M^t_{k-1} \ldots M^t_1 D^t \geq \max_{M \in \cF} R^\top M_{k-1} \ldots M_1 D^t - \varepsilon \left\Vert R \right\Vert_{\infty},
    \]
    using Algorithm~\ref{alg: max_SW}.\\
The no-regret player observes $u^t_i = \frac{R^\top M^t_{k-1} \ldots M^t_1 e_i}{\Vert R \Vert_{\infty}}$ for all $i \in [w]$, and picks $D^{t+1}$ via multiplicative weight update, as follows:
    \[
    D^{t+1}(i) = \frac{D^t(i) \beta^{u^t_i}}{\sum_{j=1}^w D^t(j) \beta^{u^t_j}}~\forall i \in [w],
    \]
 with $\beta = \frac{1}{1 + \sqrt{2 \frac{\ln w}{T}}} \in [0,1)$. 
}
\caption{2-Player Dynamics for the Ex-Ante Maximin Problem}\label{def: dynamics}
\end{algorithm}

\begin{lemma}\label{lem:exante_accuracy}
Let $T > 0$, $\overline{\Delta M}$ be the probability distribution that picks $\left(M_1,\ldots,M_{k-1}\right) \in \cF(B,M_1^0,\ldots,M_k^0)$ with probability 
\[
\frac{1}{T} \sum_{t=1}^T \mathbbm{1} \left\{\left( M_1,\ldots,M_{k-1} \right) = \left( M^t_1,\ldots,M^t_{k-1} \right) \right\},
\]
where $M^1,\ldots,M^T$ are the outputs of Algorithm~\ref{def: dynamics}. Then $\overline{\Delta M}$ $\left(\varepsilon
    +\sqrt{2 \frac{\ln w}{T}} + \frac{\ln w}{T} \right) \Vert R \Vert_{\infty}$-approximately optimizes Program~\ref{exante_maxmin_program}.
\end{lemma}

The proof of Lemma~\ref{lem:exante_accuracy} follows from interpreting Program~\ref{exante_maxmin_program} as zero-sum game, noting that the best response problem for the maximization player corresponds to the welfare-maximization problem for which we have an efficient algorithm, and then applying the no-regret dynamic analysis from~\citet{FS96}. The details are provided in  Appendix~\ref{app:exante_accuracy}.

\section{Price of Fairness}
In this section, we compute lower bounds on a notion of ``price of fairness'', and we show these lower bounds are tight when restricting attention to pipelines whose edges are \emph{all} malleable\juba{added}. Specifically, we compare the optimal welfare achievable with the welfare that is achievable if we instead use our budget to maximize the \emph{minimum} value over initial positions --- i.e. if we solve the maximin problem. We focus on the ex-post maximin problem --- i.e. we prove our bounds with respect to \emph{deterministic} solutions. We note that there may be many different maximin optimal solutions that differ in their overall welfare, and so we consider two variants of the price of fairness in our setting --- comparing with both the \emph{maximum} welfare consistent with a maximin optimal solution, and the \emph{minimum} welfare consistent with a maximin optimal solution.

Let $OPT_{SW}$ the optimal value of Program~\eqref{SW_program} (the optimal social welfare). Let $S^f$ be the set of solutions to Program~\eqref{maxmin_program} (the deterministic maximin problem). Further, define 
\[
W(M_1,\ldots,M_{k-1}) \triangleq R^\top M_{k-1} \ldots M_1 D_1^0
\]
to be the social welfare achieved by transition matrices $M_1,\ldots,M_{k-1}$, and 
\[
W_{fair}^+ \triangleq \max_{\left(M_1,\ldots,M_{k-1}\right) \in S^f} W(M_1,\ldots,M_{k-1}),
\]
\[
W_{fair}^- \triangleq \min_{\left(M_1,\ldots,M_{k-1}\right) \in S^f} W(M_1,\ldots,M_{k-1})
\]
to be the maximum and minimum social welfare respectively that are consistent with maximin optimal solutions. 
We define two variants of ``the price of fairness'' in our setting as:
\[
P_f^+ \triangleq \frac{OPT_{SW}}{W_{fair}^+} \geq 1,
\]
and
\[
P_f^- \triangleq \frac{OPT_{SW}}{W_{fair}^-} \geq 1.
\]
Note that $P_f^+ \leq P_f^-$ always, as $P_f^+$ compares the optimal social welfare with the solution of Program~\ref{maxmin_program} with highest social welfare , while $P_f^-$ considers the solution that has the lowest social welfare. We provide matching lower bounds on $P_f^+$ and upper bounds on $P_f^-$. This, in turn, provides tight bounds on the price of fairness with respect to \emph{any} choice of maximin solution. 

\subsection{Lower Bounds on $P_f^+$}\label{sec:lower_bounds_pof}

Our lower bounds are based on the following construction:
\begin{example}\label{ex: lb_construction}
Consider a network with only two layers, $L_1$ and $L_2$, such that $L_1$ has $w$ nodes and $L_2$ has $2$ nodes. Suppose the starting distribution is given by $D_1^0 = \left(1-(w-1) \varepsilon,\varepsilon,\ldots,\varepsilon\right)^\top$ for $\varepsilon > 0$ small enough, the reward vector is given by $R = (1,0)^\top$, and the initial transition matrix $M_1^0$ is given by
\begin{align*}
M_1^0 = \left(
\begin{matrix}
0 & \ldots & 0
\\1 & \ldots & 1
\end{matrix}
\right)
.
\end{align*}
I.e., in the initial transition matrix, every starting node transitions to the destination node that has reward $0$, and the welfare of the initial network is $0$. We assume all edges are malleable.
\end{example}

\begin{theorem}\label{thm: fairness_lb}
For all $w \in \mathbb{N}$, for any $\delta > 0$, there exists a network with $k=2$ with price of fairness
\begin{align*}
P_f \geq 
    \begin{cases} 
        w -\delta & \mbox{if } 0 < B \leq 2\\
        \frac{2w}{B} -\delta &\mbox{if } 2 < B \leq 2w\\
        1 &\mbox{if } B \geq 2w
    \end{cases} 
.
\end{align*}
\end{theorem}

The proof follows from solving the social welfare maximization problem and the maximin value problem on Example~\ref{ex: lb_construction}. The full proof is provided in Appendix~\ref{app:fairness_lb}.

\subsection{Upper Bounds on $P_f^-$}

Importantly, in this section, we restrict ourselves to pipelines such that \emph{all} edges are malleable. In this case, we show upper bounds that tightly match the lower bounds of Section~\ref{sec:lower_bounds_pof}.

Our upper bounds will make use of the following claim, which bounds the maximum social welfare that can be achieved under budget $B$. 

\begin{lemma}\label{lem: ub_OPT}
\[
OPT_{SW} \leq \left\Vert R \right\Vert_{\infty}
\]
and
\[
OPT_{SW} \leq  W^0 + \frac{B}{2} \left\Vert R \right\Vert_{\infty},
\]
where $W^0 = R^\top M_{k-1}^0 \ldots M_1^0 D_1^0$ is the initial welfare.
\end{lemma}

The proof of this lemma is straightforward and is deferred to Appendix~\ref{app:ub_OPT}. 

We will also need lower bounds on the social welfare achieved by any optimal solution to the maximin program. The first lower bound is a function of $B$ and $w$, but is independent of $W^0$.
\begin{lemma}\label{lem: lb_maxmin}
When all edges are malleable, for any $\left(M_1^f,\ldots,M_{k-1}^f\right) \in S^f$,
\[
W\left(M_1^f,\ldots,M_{k-1}^f\right) \geq \min\left(1,\frac{B}{2w}\right) \left\Vert R \right\Vert_{\infty}.
\]
\end{lemma}

The proof of Lemma~\ref{lem: lb_maxmin} is deferred to Appendix \ref{app:lb_maxmin}. We provide a brief proof sketch below:
\begin{proof}[Proof sketch]
The budget $B$ can be fully invested in improving edges from the second-to-last layer $L_{k-1}$ to the last layer $L_k$. The idea is to increase the transition from any node $u \in L_{k-1}$ to the best node $v \in L_k$ with reward $\Vert R \Vert_{\infty}$, by an amount of $B/2w$ each. Doing so guarantees the result, noting that every starting node in the first layer transitions to a node in $L_{k-1}$ with probability $1$, then to reward $\Vert R \Vert_{\infty}$ on the last layer $L_k$ with probability at least $B/2w$. Importantly, note that this proof relies on the fact that the edges from the second-to-last to the last layer are malleable.
\end{proof}

The second lower bound we need shows that the social welfare achieved by a solution to Program~\eqref{maxmin_program} is lower-bounded by the initial social welfare $W^0 = R^\top M_{k-1}^0 \ldots M_1^0 D_1^0$.
\begin{lemma}\label{lem: lb_maxmin_w0}
When all edges are malleable, for any $\left(M_1^f, \ldots, M_{k-1}^f\right) \in S^f$, 
\[
W\left(M_1^f, \ldots, M_{k-1}^f\right) \geq W^0.
\]
\end{lemma}

 We defer the full proof of Lemma~\ref{lem: lb_maxmin_w0} to Appendix~\ref{app:lb_maxmin_w0} and provide a proof sketch below:
 
\begin{proof}[Proof sketch]
 The proof follows from the fact that increasing the expected reward of any given node $u$ in any layer $L_t$ in the network can be done by taking some of the transition probability from $u$ to a low-reward node and re-allocating it to the transition between $u$ and a higher reward node. Doing so does not decrease the expected reward of any other vertex in the network. In turn, it is always sub-optimal to invest budget into decreasing the expected reward of any node in the network.
\end{proof}

We can now use Lemmas~\ref{lem: ub_OPT},~\ref{lem: lb_maxmin} and~\ref{lem: lb_maxmin_w0} to derive nearly tight upper bounds on the price of fairness with respect to the worst maximin solution:
\begin{theorem}
For every instance of the problem in which edges are malleable, we have that
\begin{align*}
P_f^- \leq 
    \begin{cases} 
        w + 1 & \mbox{if } 0 < B \leq 2\\
        \frac{2w}{B} &\mbox{if } 2 < B \leq 2w\\
        1 &\mbox{if } B \geq 2w
    \end{cases} 
.
\end{align*}
\end{theorem}

\begin{proof}
We divide the proof in three cases:
\begin{enumerate}
\item $B \geq 2w$. By Lemma~\ref{lem: lb_maxmin}, it must be the case that any optimal solution to Program~\eqref{maxmin_program} has welfare at least $\min\left(1,\frac{B}{2w}\right) \left\Vert R \right\Vert_{\infty} = \left\Vert R \right\Vert_{\infty}$. It is then immediately the case that $OPT_{SW} = \left\Vert R \right\Vert_{\infty}$ by Lemma~\ref{lem: ub_OPT} and $P_f = 1$.
\item $2 < B \leq 2w$. By Lemma~\ref{lem: ub_OPT}, we have $OPT_{SW} \leq \left\Vert R \right\Vert_{\infty}$. Further, by Lemma~\ref{lem: lb_maxmin}, we have that any solution to Program~\eqref{maxmin_program} has welfare at least $\frac{B}{2w} \left\Vert R \right\Vert_{\infty}$. This immediately yields the result.
\item $0 < B \leq 2$. By Lemma~\ref{lem: ub_OPT}, we have $OPT_{SW} \leq W^0 + \frac{B}{2} \left\Vert R \right\Vert_{\infty}$. By Lemmas~\ref{lem: lb_maxmin} and~\ref{lem: lb_maxmin_w0}, we have that the social welfare of any maximin solution is at least $W^0$ and at least $\frac{B}{2w} \left\Vert R \right\Vert_{\infty}$. Therefore, the price of fairness is upper-bounded on the one hand by 
\[
P_f^- \leq \frac{W^0 + \frac{B}{2} \left\Vert R \right\Vert_{\infty}}{W^0} = 1 + \frac{\frac{B}{2} \left\Vert R \right\Vert_{\infty}}{W^0}.
\]
and on the other hand by 
\[
P_f^- \leq \frac{W^0 + \frac{B}{2} \left\Vert R \right\Vert_{\infty}}{\frac{B}{2w} \left\Vert R \right\Vert_{\infty}} = w + \frac{W^0}{\frac{B}{2w} \left\Vert R \right\Vert_{\infty}}.
\]
When $W^0 \geq \frac{B}{2w} \left\Vert R \right\Vert_{\infty}$, the first bound gives
\[
P_f^- \leq 1 + \frac{\frac{B}{2} \left\Vert R \right\Vert_{\infty}}{\frac{B}{2w} \left\Vert R \right\Vert_{\infty}} = w + 1,
\]
and when $W^0 \leq \frac{B}{2w} \left\Vert R \right\Vert_{\infty}$, the second bound yields
\[
P_f^- \leq w + \frac{\frac{B}{2w} \left\Vert R \right\Vert_{\infty}}{\frac{B}{2w} \left\Vert R \right\Vert_{\infty}} = w + 1,
\]
which concludes the proof.
\end{enumerate}
\end{proof}

\section{Hardness of Approximation}

In this section, we show that the problem of finding the ex-post maximin value of a pipeline intervention problem instance within an approximation factor of $2$ is $\nph$-hard in the general case, where the width $w$ of the network is not bounded. More specifically, we show that no algorithm that has a time bound polynomial in $w,k$ and $B$ can give a $2$-approximation to the maximin value unless $P=NP$. This hardness result holds for $k$ as small as $17$. We remark that our result and proof can be immediately extended to show hardness of $C$-approximation, for any constant $C$, for an appropriate choice of constant depth $k$.

% The problem that we show to be hard to approximate is a slight generalization of our original problem. We have assumed up to now that we are allowed to modify the probabilities on all the edges, whereas, in our generalization, only a subset of the edges between adjacent layers are ``malleable'' and the remaining are fixed; probability mass can now only be moved between malleable edges. We remind the reader that our algorithmic results generalize to this setting, as per Remark~\ref{rem:simplifications}.

We show this hardness result via a reduction from a gap version of the vertex cover problem. The result of~\cite{dinur2005hardness} shows that it is $\nph$-hard to approximate the minimum vertex cover by a factor smaller than $1.306$. In particular, their result shows that the following gap version of vertex cover is $\nph$-hard: given $(\cG,\kappa)$, we wish to either know if the graph $\cG$ has a vertex cover of size $\kappa$, or has no vertex cover smaller than size $1.306 \kappa$. 

%We provide our reduction as well as a formal statement of our hardness result below, and defer the proof to Appendix~\ref{app:hardness}.
\paragraph{The Reduction}

Our reduction works as follows: we construct a pipeline intervention instance of constant width (17 layers) from the given graph. The first layer has a node corresponding to each edge $(u,v)$ of the original graph, and is connected by edges to nodes corresponding to vertices $u$ and $v$ on the second layer. We set up the instance so that positive probability mass is only ever added to a set of edge disjoint paths, where each path corresponds to a vertex in the original graph. These paths are shown by the dark, solid lines in Figure~\ref{fig:hardness_instance}. The main idea behind the reduction is the following - by observing how allocations finding the maximin value split the budget over these edge disjoint paths, we can find out which vertices would form a small vertex cover of the original graph. 

\begin{figure}[h]
    \centering
    \includegraphics[width=10cm]{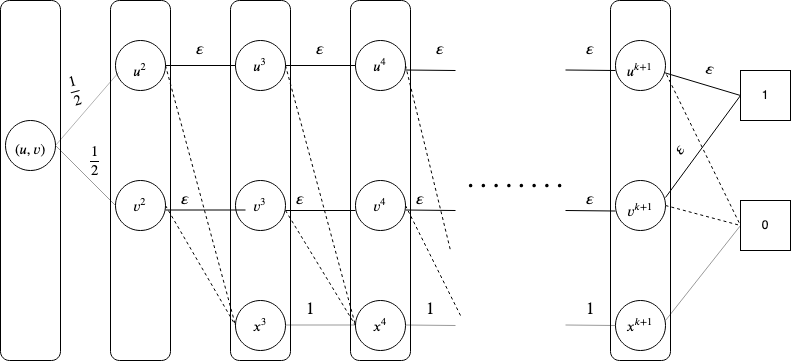}
    \caption{Constructed Instance of the Pipeline Intervention problem}
    \label{fig:hardness_instance}
\end{figure}

Formally, let the given graph $\cG=(\cV,\cE)$ we reduce from have $n$ vertices ($|\cV| =n $) and $m$ edges ($|\cE| = m$).  We construct of a pipeline intervention problem instance $\cI'$ with $k+2$ layers and width $w$, where $k =15$ and $w$ is polynomial in $n$. The instance $\cI'$ has an associated budget $B(\kappa, \varepsilon) = 2 k\kappa \varepsilon$ where $\varepsilon < \frac{1}{2}$. For the sake of clarity, we refer to the set of vertices $\cV$ in the vertex cover instance as ``vertices'' and the vertices in the instance $\cI'$ as ``nodes''. A complete description of instance $\cI'$ is as follows:
\begin{enumerate}
     \item The first layer, $L_1$, has exactly $m$ nodes, with each edge $(u,v)$ in graph $\cG$ having a unique corresponding node of the same label in layer $L_1$. 
    \item The second layer has exactly $n$ nodes, with each vertex $v$ in graph $\cG$, having a unique corresponding node in layer $L_2$ with label $v^2$.
    \item The next $k-1$ layers are of the following form - layer $L_i$, for $i=3$ to $k+1$, has $n+ 1$ nodes. The first $n$ nodes have labels from the set $\{v^i\}_{v \in V}$, i..e, each vertex $v$ in the original graph $\cG$ has a corresponding node $v^i$ in layer $L_i$. The last node is indexed by $x^i$ and exists to capture the ``leftover'' outward probability from the nodes $\{v^{i-1}\}_{v \in  V}$ in layer $L_{i-1}$.
    \item The final layer $L_{k+2}$ has two reward nodes - $y$, of reward $1$ and $z$, of reward $0$. 
\end{enumerate}  
We now describe the initial transition matrices.
\begin{enumerate}
    \item From layer $L_1$ to layer $L_2$: for every node $(u,v)$ in layer $L_1$, the outgoing probability is equally split between edges to nodes $u^1$ and $v^1$ in layer $L_2$, i.e., edges $((u,v),u^1)$ and $((u,v),v^1)$ each have probability $\frac{1}{2}$.
     \item From layer $L_i$ to layer $L_{i+1}$ for $i=2$ to $k$: For all vertices $v \in \cV$ (i.e., the original graph), the corresponding edge $(v^i,v^{i+1})$ (in our construction) has probability $\varepsilon$. The remaining outgoing probability out of node $v^i$ goes to the leakage node $x^{i+1}$. We call edges of the form $(v^i, x^{i+1})$ ``leakage'' edges. For $i \geq 3$, the edge $(x^i,x^{i+1})$ has all the outward probability, i.e., $1$,  from node $x^i$.
    \item From layer $L_{k+1}$ to layer $L_{k+2}$: each node in layer $L_{k+1}$ is connected to $z$, the zero reward node, with probability $1$.
\end{enumerate}

We let $P_v$ be the path going through nodes $v^2,v^3 \cdots v^{k+1}, y$ in our construction. We will refer to $\{P_v\}_{v \in \cV}$ as vertex paths. Let $E'$ be the set of edges found on paths $\{P_v\}_{v \in V}$. Let $E''$ contain of all the ``leakage'' edges in the instance $\cI'$ ,i.e., edges of the form $(v^i, x^{i+1})$ as well as all edges of the form $(v^{k+1},z)$. We stipulate, as part of the description of the instance, that $E' \cup E''$ is the set of malleable edges in $\cI'$ and that the probability mass on any other edge cannot be changed. This completes the description of instance $\cI'$.

\paragraph{Formal Hardness Result}
In this setup, we argue that any algorithm computing a maximin value of our instance will only use budget to improve the probability mass on the vertex paths. To begin, we observe that the only nodes with more than one outgoing malleable edge are the nodes on the vertex paths, exempting the reward node $y$. Hence, these are the only nodes whose associated outgoing transition matrices can be modified. Furthermore, these nodes have exactly one edge from $E'$, the edges on the vertex paths, and one leakage edge, implying that probability mass can only be transferred between these two edges. Note that the reward associated with starting at any leakage node is always $0$, since all outgoing edges from leakage nodes lead to the zero sink with probability $1$, and are not malleable.  
 %and no other outgoing edges are present in the malleable edge set $E' \cup E''$. 
 In contrast, nodes on the vertex paths start with expected positive welfare initially. Therefore, it would be a (strictly) sub-optimal strategy to route any more probability mass towards a leakage node , i.e., increase the probability mass on any edge in $E''$ leading to reward $0$, since it would come at the cost of decreasing probability mass on an edge towards a node on a vertex path with strictly positive reward. Consequently, any algorithm that optimizes the maximin value would only increase the probability on the other malleable edges, i.e., edges in $E'$, by removing probability mass from edges $E''$. To give a more local picture, budget is spent on the following operation - increasing the probability mass on an edge of the form $(v^i,v^{i+1})$ and balancing the outgoing probability mass from vertex $v^i$ by correspondingly decreasing the probability mass on edge $(v^i,x^{i+1})$. Note that there are no malleable edges in the transition between layer $L_1$ and layer $L_2$, and hence the corresponding transition matrix remains unchanged.

We fix a threshold $T \triangleq \frac{(2 \varepsilon)^k}{4} $. The following lemma shows the desired reduction:
\begin{lemma}
\label{lemma:keyreduction}
If graph $\cG$ has a vertex cover of size $\kappa$, then the maximin value of the constructed instance is at least $2T$. Complementarily, if graph $\cG$ has no vertex cover smaller than $1.3605 \kappa$, then the maximin value of the constructed instance is less than $T$.
\end{lemma}
Observe that the above lemma shows achieving a $2$-approximation to the maximin value is $\nph$-hard, since such an algorithm would be able to solve the gap version of vertex cover in polynomial time. Note that this also rules out an additive PTAS for the pipeline interventions problem, even for networks of constant depth , since our hard instance $\cI'$ has constant depth.

\begin{proof}
The forward direction of the proof is straightforward. Given that the original graph has a vertex cover $\cV^* \subseteq \cV$ of size $\kappa$, we show a candidate solution that guarantees a minimum reward of $2T$. Consider the subset $E^* \subseteq E$ which consists only of edges on the vertex paths indexed by vertices in the vertex cover $\cV^*$, i.e., $\{P_v\}_{v \in \cV^*}$. Since each such path has $k$ edges and there are $\kappa$ number of paths, $|E^*| = k \kappa$.
We spend a budget $2 \varepsilon$ to increase the probability mass one each edge of $E^*$ by $\varepsilon$ and decrease the probability mass of the corresponding leakage edge by $\varepsilon$. Note that this is always possible, since $\varepsilon + \varepsilon \le 1$ (i.e., we never try to increase the probability mass on an edge beyond $1$). This exactly utilizes our budget $B(\kappa, \varepsilon) = 2k \kappa \varepsilon$. The new probability on every edge of $E^*$ is now $2\varepsilon$. Consider any population starting on a node $(u,v)$ of the first layer. Since $\cV^*$ is a vertex cover, at least one of $u$ or $v$ is present in the set $\cV^*$. Without loss of generality, let us assume $v \in \cV^*$. The reward at node $v^2$ is now exactly $(2\varepsilon)^k$. Since the transition matrices for layer $L_1$ to $L_2$ are unchanged, a population starting at node $(u,v)$ arrives at node $v^2$ with probability $\frac{1}{2}$; thus, an agent starting at $(u,v)$ has expected reward at least $\frac{(2 \varepsilon)^k}{2} = 2T$.

Now, consider the case where the graph has no vertex cover of size $1.305 \kappa$. To prove a contradiction, let us assume that the maximin value is greater than or equal to $T$. We will show that we can recover a vertex cover of size less than size $1.305 \kappa$ from the maximin value solution, thus proving a contradiction. Let $W(l)$ represent the reward associated with starting at node $l$. We know that all the budget is spent on the paths $\{P_v\}_{v \in V}$. Let use write $2k \varepsilon a_v$ be the budget spent on improving path $P_v$; note that the net increase in probability mass across all edges of $P_v$ is $k \varepsilon a_v$. By the AM-GM inequality, we know that $W(v^2) \le (\varepsilon + a_v \varepsilon)^k$, with equality when $k a_v \varepsilon$ is split equally across edges on path $P_v$. A population starting at node $(u,v)$ on layer $L_1$ reaches node $u_2$ and node $v^2$ each with probability $\frac{1}{2}$ each. Thus, we have $W((u,v)) \le \frac{(\varepsilon + a_u \varepsilon)^k + (\varepsilon + a_v \varepsilon)^k}{2}$. Since this solution guarantees a reward of at least $T$ for every vertex within a total budget $B(\kappa, \varepsilon) = 2 k \kappa \varepsilon$, we have the following inequalities:
\begin{align*}
    \frac{(\varepsilon + a_u \varepsilon)^k + (\varepsilon + a_v \varepsilon)^k}{2} &\ge \frac{(2\varepsilon)^k}{4} \quad \forall (u,v) \in E
    \\ \sum_{v \in V} a_v \le \kappa \varepsilon
\end{align*}
Dividing the first inequality by $\varepsilon^k$, multiplying by $2$, and substituting $k =15$, we get:
\begin{align*}
 (1 + a_u)^{15} + (1 + a_v)^{15} &\ge 2^{14} \quad \forall (u,v) \in E
    \\ \sum_{v \in V} a_v \le \kappa
\end{align*}
We now generate a vertex set $\cV^*$ of original graph $\cG$ as follows - include vertex $v$ in $\cV^*$ if $a_v \ge 0.823$. We complete the proof by showing that $\cV^*$ is a vertex cover of $G$ and $|\cV^*| < 1.3605 \kappa$. The fact that $|\cV^*| < 1.3605 \kappa $ follows directly from the rounding scheme that is employed. Observe that each $a_v$ is scaled upward by a factor of at most $\frac{1}{0.823} < 1.22$. Thus, $|\cV^*| < 1.22 \kappa$. Now, assume that $\cV^*$ is not a vertex cover. Then, there exists an edge $(u,v) \in \cE$ in the original graph such that $a_u, a_v < 0.823$. Thus, we get:
\begin{align*}
    (1 + a_u)^{15} + (1 + a_v)^{15} &< 2 (1.823)^{15} < 16325< 16384 =   2^{14}
\end{align*}
This violates the reward guarantee for the population starting at node $(u,v)$, thus resulting in a contradiction.
\end{proof}

\bibliographystyle{plainnat}
\bibliography{bibliography.bib}

\appendix
\section{Preliminaries, continued}
\subsection{Proof of Claim~\ref{clm: budget_loss}}
\label{app:budget_loss}

We will prove both results at once, noting that they both directly follow from showing that there exists matrices $M_1^\varepsilon,\ldots,M_k^\varepsilon$ that are feasible for the discretized problems, such that for all $j \in [w]$,
\[
R^\top M_{k-1}^\varepsilon \ldots M_{1}^\varepsilon e_j \geq R^\top M^*_{k-1} \ldots M^*_{1} e_j - (k-1) \varepsilon,
\] 
where $M^*_1,\ldots,M^*_{k-1}$ is an optimal solution to Program~\eqref{SW_program}, respectively~\eqref{maxmin_program}. To do so, we first note that it is feasible for Program~\eqref{dscrt_SW_program} to pick a split of the budget $B_1^\varepsilon,\ldots,B_{k-1}^\varepsilon$ such that for all $t$, $B_t^\varepsilon \geq \max(c(M_t^*,M_t^0) - \varepsilon,0)$, by construction of $\cB(\varepsilon)$. We are going to construct a matrix $M^\varepsilon_t$ that is close to $M_t^*$ and requires budget at most $B_t^\varepsilon$. 

When $M_t^* = M_t^0$, one can just let $M^\varepsilon_t = M_t^0$. Now, suppose $c(M_t^*,M_t^0) > 0$. For every pair of nodes $u \in L_{t}$, we let $S_u^+$ the set of vertices $v \in L_{t+1}$ such that $M_t^*(v,u) > M_t^0(v,u)$ (i.e. the transition from $u$ to $v$ has higher probability in $M_t^*$ than in $M_0$), and $S_u^-$ the set of vertices $v \in L_{t+1}$ such that $M_t^*(v,u) < M_t^0(v,u)$. We note immediately that 
\[
c(M_t^*,M_t^0) = \sum_{u \in L_t} \left(
\sum_{v \in S_u^+} \left(M_t^*(v,u) - M_t^0(v,u)\right) + \sum_{v \in S_u^-} \left(M_t^0(v,u) - M_t^*(v,u)\right),
\right)
\]
Now, let us construct $M^\varepsilon_t \in \cM$ such that for every $u$, 
\[
M^\varepsilon_t(v,u) = M_t^*(v,u) - \alpha(v,u)~\forall v \in S_u^+
\]
and
\[
M^\varepsilon_t(v,u) = M_t^*(v,u) + \alpha(v,u)~\forall v \in S_u^-,
\]
where $\alpha(v,u) \geq 0$ for all $u,v$, $\sum_{u,v} \alpha(v,u) = \min(\varepsilon,c(M_t^*,M_t^0))$, $\sum_{v \in S_u^+} \alpha(v,u) = \sum_{v \in S_u^-} \alpha(v,u)$, and $M^\varepsilon_t(u,v) \geq M_t^0(u,v)$ for $v \in S_u^+$ and $M_t^\varepsilon (u,v) \leq M_t^0(u,v)$ for $v \in S_u^-$. Note that such $\alpha$'s exist by virtue of $\min(\varepsilon,c(M_t^*,M_t^0)) \leq c(M_t^*,M_t^0)$, which is the absolute value amount by which $M_t^*$ differs from $M_t^0$ coordinate-by-coordinate. Second, note that only malleable edges $(u,v)$ have $M_t^\varepsilon(v,u) \neq M_t^0(v,u)$, since we only modify malleable edges where $M_t^*(v,u) \neq M_t^0(v,u)$. Further, $M_t \in \cM$ since all the coefficients of $M_t$ remain between $0$ and $1$, and for all $u$,
\[
\sum_v M_t(v,u) = \sum_v M_t^*(v,u) + \sum_{v \in S_u^-} \alpha(v,u) - \sum_{v \in S_u^+} \alpha(v,u) = \sum_v M_t^*(v,u) = 1.
\]
Further, the cost of moving from $M_t^0$ to $M_t^\varepsilon$ is given by
\begin{align*}
c(M_t^\varepsilon,M_t^0) 
&= \sum_{u \in L_t} \left(
\sum_{v \in S_u^+} \left(M_t^\varepsilon(v,u) - M_t^0(v,u)\right) + \sum_{v \in S_u^-} \left(M_t^0(v,u) - M_t^\varepsilon(v,u)\right)
\right)
\\&= \sum_{u \in L_t} \left(
\sum_{v \in S_u^+} \left(M_t^*(v,u) - M_t^0(v,u) -\alpha(v,u) \right) + \sum_{v \in S_u^-} \left(M_t^0(v,u) - M_t^*(v,u) - \alpha(v,u)\right)
\right)
\\&= c(M_t^*,M_t^0) - \sum_{u,v} \alpha(v,u)
\\&= \max\left(0,c(M_t^*,M_t^0) - \varepsilon\right),
\end{align*}
noting that if $v$ was in $S_u^+$ (resp $S_u^-$), it is still the case that $M_t^\varepsilon(v,u) \geq M_t^0(v,u)$ (resp. $M_t^\varepsilon(v,u) \leq M_t^0(v,u)$). In turn, $M_t^\varepsilon$ requires at most budget $B_t^\varepsilon$, and $M_1^\varepsilon,\ldots,M_{k-1}^\varepsilon$ is a feasible solution for the discretized programs. Finally, for any transition matrices $M_1,\ldots,M_{k-1}$, letting $R_{t+1}^\top = R^\top M_{k-1} \ldots, M_{t+1}$ (trivially, $0 \leq R_{t+1} \leq \left\Vert R \right\Vert_{\infty}$) and $D_{t,j} = M_{t-1} \ldots M_1 e_j$ (trivially, $D_{t,j} \in \cD$), we note that (letting $\alpha(v,u) = 0$ where not defined) 
\begin{align*}
R_{t+1}^\top M_t D_{t,j}
&=\sum_{u \in L_t,~v \in L_{t+1}} M_t^\varepsilon(v,u) R_{t+1}(v) D_{t,j}(u)
\\&\geq \sum_{u,v} M_t^*(v,u) R_{t+1}(v) D_{t,j}(u)
- \sum_{u,v} \alpha(v,u) R_{t+1}(v) D_{t,j}(u)
\\&\geq \sum_{u,v} M_t^*(v,u) R_{t+1}(v) D_{t,j}(u)
-\left\Vert R \right\Vert_{\infty} \sum_{u,v} \alpha(v,u)
\\&\geq \sum_{u,v} M_t^*(v,u) R_{t+1}(v) D_{t,j}(u)
-\left\Vert R \right\Vert_{\infty} \varepsilon,
\end{align*}
where the first inequality uses that for all $u,v$, $M_t^\varepsilon(v,u) \geq M_t^*(u,v) - \alpha(v,u)$ by construction, the second inequality that $0 \leq \alpha(v,u)$, $0 \leq D_{t,j}(u) \leq 1$ and $0 \leq R_{t+1}(v) \leq \left\Vert R \right\Vert_{\infty}$, and the last inequality from the fact that $\sum_{u,v} \alpha(v,u) = \min(\varepsilon,c(M_t^*,M_t^0)) \leq \varepsilon$. The proof can be concluded noting that for all $t$, the above inequality implies
\begin{align*}
R^\top M_{k-1}^* \ldots M_{t+1}^* M_t^\varepsilon \ldots M_1^\varepsilon e_j
\geq R^\top M_{k-1}^* \ldots M_{t+2}^* M_{t+1}^\varepsilon M_t^\varepsilon \ldots M_1^\varepsilon e_j -\left\Vert R \right\Vert_{\infty} \varepsilon
\end{align*}
and applying this new inequality recursively. 

\subsection{Proof of Claim~\ref{clm:eps_net}}\label{app:eps_net}

This is a well-known result; we provide a proof for completeness. The first observation is that such a net can be constructed recursively, as follows. Start with an empty set $S$. Initialize by picking any point $D$ in $\cD$, and let $S = \{D\}$. Then, recursively keep finding points $D' \in \cD$ such that for all $D \in S$, $\left\Vert D - D' \right\Vert_1 > \varepsilon$, and augment $S := S \cup D'$. Finally, stop the algorithm when no such point $D' \in \cD$ exists. By construction, it must be that when the algorithm terminates, for all $D' \in \cD$, there exists $D \in S$ with $\left\Vert D - D' \right\Vert_1 \leq \varepsilon$. As such, $S$ constitutes an $\varepsilon$-net in $\ell_1$ distance for $\cD$.

 Second, we bound the number of steps needed. To do so, we remark that by construction, for all $D_1,D_2 \in S$, it must be the case that $\left\Vert D_1 - D_2 \right\Vert_1 > \varepsilon$; in turn, the $\ell_1$-balls of radius $\varepsilon/2$ around each of the elements of $S$ must be disjoint, and the sum of their volumes is less than the volume of $\cD$. Since the volume of the probability simplex is given by $\frac{1}{w!}$, and the volume of an $\ell_1$-ball of radius $r$ is given by $\frac{1}{w!} \left(2r\right)^w$, this yields that $|S| \times \frac{\varepsilon^w}{w!} \leq \frac{1}{w!}$, or equivalently $|S| \leq \left(\frac{1}{\varepsilon}\right)^w$.

\section{A More General Cost Model}\label{app:Lipschitz}

In this section, we show how to extend our algorithmic results to convex costs whose variations are lower-bounded. More specifically, we make the following assumptions on $c(M_t,M_t^0)$, the cost of transforming the initial transition matrix $M_t^0$ into an alternative $M_t \in \cM$. 

\begin{assumption}[Initial Condition]\label{as:zero}
$c(M_t^0,M_t^0) = 0$.
\end{assumption}
This encodes the natural assumption that not intervening on the transition matrix incurs no cost. 

\begin{assumption}[Convexity]\label{as:convex}
The function $M_t \rightarrow c(M_t,M_t^0)$ is convex for all $M_t^0 \in \cM$.
\end{assumption}
This is a standard assumption, that ensures the optimization problem solved for a single layer is convex and efficiently solvable.

\begin{assumption}[Linearly increasing costs]\label{as:linear_increase_cost}
% Further, if for all $(i,j) \in [w^2]$, either $M_t^2(i,j) \leq M_t^1(i,j) \leq M_t^0(i,j)$ or $M_t^2(i,j) \geq M_t^1(i,j) \geq M_t^0(i,j)$, the following two properties must hold:
% \begin{enumerate}
% \item Monotonicity from $M_t^0$: $c(M_t^1,M_t^0) \leq c(M_t^2,M_t^0)$.\juba{This is in fact just a consequence of convexity and maybe should be written as such}
% \item Anti-Lipschitzness: 
% There exists a constant $L$ such that for all $M_t^0,~M_t^1,~M_t^2 \in \cM$ that satisfy the condition of (2),
% \[
% \left\vert c(M_t^2,M_t^0) - c(M_t^1,M_t^0) \right\vert
% \geq L  \sum_{(i,j) \in [w]^2} \left\vert M_t^1(i,j) - M_t^2(i,j)\right\vert.
% \]
%\end{enumerate}
There exists a constant $L > 0$ such that
\[
\left\vert c(M_t^2,M_t^0)  - c(M_t^1,M_t^0) \right\vert
\geq L  \sum_{(i,j) \in [w]^2} \left\vert M_t^2(i,j) - M_t^1(i,j)\right\vert
\]for any $M_t^0,~M_t^1,~M_t^2 \in \cM$ such that for all $(i,j) \in [w^2]$, either $M_t^2(i,j) \leq M_t^1(i,j) \leq M_t^0(i,j)$ or $M_t^2(i,j) \geq M_t^1(i,j) \geq M_t^0(i,j)$. 
\end{assumption}
%  The first assumption formalizes the natural property that if $M_t^2(i,j)$ is further away from $M_t^0(i,j)$ than $M_t^1(i,j)$ for all transitions $(i,j)$, it should be the case that it is costlier to mover from $M_t^0$ to $M_t^2$ than from $M_t^0$ to $M_t^1$. The second properties guarantees that improving transitions does not come for free: greater improvements lead to higher costs. Both assumptions only apply when $M_t^1$ and $M_t^2$ have been modified in the same direction starting from $M_t^0$, as otherwise the costs of $M_t^1$ and $M_t^2$ may become incomparable in practice. 
This assumption formalizes the natural property that modifying $M_t^0$ does not come for free: the further away $M_t$ is from $M_t^0$, the higher the cost of the transformation from $M_t^0$ to $M_t$. The assumption requires that an increase of the distance between $M_t(i,j)$ and $M_t^0(i,j)$ by one translates into at least a $L$ increase in the incurred cost. The condition on $M_t^1$ and $M_t^2$ implies that $M_t^2$ is further away from $M_t^0$ than $M_t^1$ by an amount given exactly by $\sum_{(i,j) \in [w]^2} \left\vert M_t^2(i,j) - M_t^1(i,j)\right\vert$.

Note that this condition on $M_t^1$ and $M_t^2$ fixes the direction --- defined by the set of edges that are increased and the set of edges that are decreased --- in which we move from $M_t^0$ to $M_t$, and prevents comparisons between matrices that have been obtained by changing $M_t^0$ in \emph{different} directions. Such modifications in different directions can be incomparable in practice, which is why we make no assumption on how they compare in terms of cost. Note that the cost function $c(M_t,M_t^0)$ used in the main body of the paper immediately satisfies Assumptions~\ref{as:zero},~\ref{as:convex} (as a sum of convex functions), and~\ref{as:linear_increase_cost} (with equality for $L = 1$).

To argue that our results carry through to cost functions that satisfy Assumptions~\ref{as:zero},~\ref{as:convex},~\ref{as:linear_increase_cost}, we first note that by the same proof as that of Theorem~\eqref{alg: max_SW}, Algorithms~\eqref{alg: max_SW} and~\eqref{alg: max_MM} find solutions with social welfare at least $OPT^\varepsilon - 2(k - t) \varepsilon \Vert  R \Vert_{\infty}$ and maximin value at least $OPT_{MM}^\varepsilon - 2(k - t) \varepsilon \Vert  R \Vert_{\infty}$. Further, since the costs are convex, Programs~\eqref{SW_DP} and~\eqref{MMW_DP} are convex optimization programs that can be solved in polynomial running times $f(w),~g(w)$. The total running times of Algorithms~\eqref{alg: max_SW} and~\eqref{alg: max_MM} are otherwise unchanged.

To show that Algorithms~\eqref{alg: max_SW} and \eqref{alg: max_MM} efficiently find near-optimal solutions to Programs~\eqref{SW_program} and \eqref{maxmin_program}, it is therefore enough to bound the difference between $OPT^\varepsilon$ (resp. $OPT_{MM}^\varepsilon$) and $OPT$ (resp. $OPT_{MM}$). We do so in Claim~\ref{clm: lipschitz_budget_loss} below:
\begin{claim}\label{clm: lipschitz_budget_loss}
There exists a feasible solution $\left(M_1^\varepsilon,\ldots,M_{k-1}^\varepsilon\right)$ to Program~\eqref{dscrt_SW_program} (resp. Program~\eqref{dscrt_maxmin_program}) with objective value at least $OPT_{SW} - \frac{(k-1) \varepsilon}{L} \left\Vert R \right\Vert_{\infty}$ 
(resp. $OPT_{MM} - \frac{(k-1) \varepsilon}{L} \left\Vert R \right\Vert_{\infty}$).
\end{claim}

\begin{proof}
We show that there exists matrices $M_1^\varepsilon,\ldots,M_k^\varepsilon$ that are feasible for the discretized problem of Program~\eqref{dscrt_SW_program}, such that for all $j \in [w]$, $R^\top M_{k-1}^\varepsilon \ldots M_{1}^\varepsilon e_j \geq R^\top M^*_{k-1} \ldots M^*_{1} e_j - (k-1) \varepsilon$, where $M^*_1,\ldots,M^*_{k-1}$ is an optimal solution to Program~\eqref{SW_program}, respectively~\eqref{maxmin_program}. We let $B_t = c(M_t^*,M_t^0)$ for simplicity of notations. We also write $\varepsilon' = \frac{\varepsilon}{L}$, and consider a budget split for Program~\eqref{dscrt_SW_program} such that for all $t$, $B_t^\varepsilon \geq \max(B_t - \varepsilon,0)$. Such a budget split is feasible by construction of $\cB(\varepsilon)$.

We now construct $M_t^\varepsilon$ as follows: we let $M_t^\varepsilon \triangleq \lambda M_t^0 + (1-\lambda) M_t^*$ where $\lambda = \frac{\varepsilon'}{\sum_{i,j} \left\vert M_t^*(i,j) - M_t^0(i,j)\right\vert}$ if $\varepsilon' \leq \sum_{i,j} \left\vert M_t^*(i,j) - M_t^0(i,j)\right\vert$, and $\lambda = 1$ otherwise. In the first case, note that since $\lambda \in [0,1]$, it must be that for all $(i,j)$, either $M_t^*(i,j) \geq \lambda M_t^0(i,j) + (1-\lambda) M_t^*(i,j) \geq M_t^0(i,j)$ or $M_t^*(i,j) \leq \lambda M_t^0(i,j) + (1-\lambda) M_t^*(i,j) \leq M_t^0(i,j)$. We can therefore apply Assumption~\ref{as:linear_increase_cost} to show that 
\begin{align*}
    c(M_t^\varepsilon,M_t^0) 
    &\leq c(M_t^*,M_t^0) - L \sum_{i,j} \left\vert M_t^*(i,j) - M_t^\varepsilon(i,j)\right\vert 
    \\&= B_t - L \sum_{i,j} \left\vert M_t^*(i,j) - \lambda M_t^0(i,j) - (1-\lambda) M_t^*(i,j) \right\vert
    \\&= B_t - L \lambda \sum_{i,j} \left\vert M_t^*(i,j) - M_t^0(i,j)\right\vert
    \\&= B_t - L \varepsilon'
    \\&= B_t - \varepsilon.
\end{align*}
Therefore, $M_t^\varepsilon$ is feasible for budget $B_t^\varepsilon$. In the second case, $\lambda = 1$, hence $M_t^\varepsilon = M_t^0$ and is feasible for budget $0$ hence $B_t^\varepsilon$.

It remains to show that $M_t^\varepsilon$ yields a good approximation to $M_t$. To see this, for any transition matrices $M_1,\ldots,M_{k-1}$, let $R_{t+1}^\top = R^\top M_{k-1} \ldots, M_{t+1}$ (trivially, $0 \leq R_{t+1} \leq \left\Vert R \right\Vert_{\infty}$) and $D_{t,j} = M_{t-1} \ldots M_1 e_j$ (trivially, $D_{t,j} \in \cD$). We note that 
\begin{align*}
R_{t+1}^\top (M_t^*-M_t^\varepsilon) D_{t,j}
&= \lambda R_{t+1}^\top (M_t^*-M_t^0) D_{t,j}
\\&= \lambda \sum_{u,v} R_{t+1}(v) (M_t^*(v,u)-M_t^0(v,u)) D_{t,j}(u)
\\&\leq \lambda \Vert R \Vert_{\infty} \sum_{u,v} \left\vert M_t^*(v,u)-M_t^0(v,u) \right\vert.
\end{align*}
There are now two cases. Either i) $M_t^* = M_t^0$, in which case $M_t^* = M_t^\varepsilon$ hence $R_{t+1}^\top (M_t^*-M_t^\varepsilon) D_{t,j} = 0$, or ii) $\lambda = \frac{\varepsilon'}{\sum_{u,v} \left\vert M_t^*(v,u)-M_t^0(v,u) \right\vert}$, and we obtain $R_{t+1}^\top (M_t^*-M_t^\varepsilon) D_{t,j} \leq \Vert R \Vert_{\infty} \varepsilon' = \Vert R \Vert_{\infty} \frac{\varepsilon}{L}$ by the above equation. In turn, for all $t$, $R_{t+1}$, and $D_{t,j}$, we have that 
$
R_{t+1}^\top M_t^* D_{t,j} -  \Vert R \Vert_{\infty} \frac{\varepsilon}{L} \leq R_{t+1}^\top M_t^\varepsilon D_{t,j}
$. 
Therefore, we have that
\begin{align*}
R^\top M_{k-1}^* \ldots M_{t+1}^* M_t^\varepsilon \ldots M_1^\varepsilon e_j
\geq R^\top M_{k-1}^* \ldots M_{t+2}^* M_{t+1}^\varepsilon M_t^\varepsilon \ldots M_1^\varepsilon e_j -\left\Vert R \right\Vert_{\infty} \frac{\varepsilon}{L},
\end{align*}
and the result can be obtained via a straightforward induction on $t$.
\end{proof}

\section{A Separation between Ex-ante and Ex-post Maximin Values}\label{app:separation}
In this section, we show there is a separation between ex-ante and ex-post maximin  welfare. We do so by constructing a specific instance $\cI$ for which the ex-ante maximin value is strictly larger than the ex-post maximin value:

\begin{figure}[h]
    \centering
    \includegraphics[width=10cm]{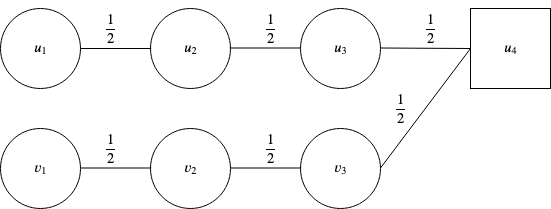}
    \caption{Part $\cI_1$ of Instance $\cI$. Edges not explicitly drawn have transition probability $0$.
    }
    \label{fig:separation_instance}
\end{figure}

We begin by giving a complete description of our instance $\cI$. The instance consists of two parts, $\cI_1$ (depicted in Figure~\ref{fig:separation_instance}) and $\cI_2$. Note that the transition probabilities from any given node in $\cI_1$, shown in Figure~\ref{fig:separation_instance}, do not sum to $1$. Part $\cI_2$ of the instance,  described below, serves to complete these transitions and ensure that the total outgoing probability of any node is $1$. We assume every edge starting from a node in $\cI_1$ is malleable (this includes edges pointing to $\cI_2)$, and every edge starting from a node in $\cI_2$ is non-malleable (this includes edges pointing to $\cI_1$).

Our proof primarily focuses on part $\cI_1$, as we argue that a centralized designer should only ever invest his budget into increasing the weight on edges with both ends in $\cI_1$ (as long as the budget is not too big).%, as long as the budget $B$ is small enough.

Formally, the instance consists of four layers ($L_1$ to $L_4$) and a small enough total budget $B$. The vertices $u_1$ to $u_4$ and $v_1$ to $v_3$, combined with all edges between them can be thought of as the part $\cI_1$ of instance $\cI$. All the remaining vertices and edges are considered to be part of $\cI_2$.
\begin{enumerate}
    \item The first layer consists only of vertices $u_1$ and $v_1$.
    \item The second layer consists of $3$ vertices in total. There are two special vertices -  $u_2$ and $v_2$, seen in Figure~\ref{fig:separation_instance}. The remaining vertex is called $x$ and is part of $\cI_2$.
    \item The third layer consists of $3$ vertices in total. There are two special vertices -  $u_3$ and $v_3$, seen in Figure~\ref{fig:separation_instance}. The remaining vertex is labeled $y$ and is part of $\cI_2$.
    \item The fourth layer, which is the reward layer, consists of two vertices, vertex $u_4$ of reward $1$ and vertex $z$, of reward $0$.
\end{enumerate}
The initial transition matrices are given as follows: 
\begin{enumerate}
    \item From layer $L_1$ to layer $L_2$ - the edge $(u_1,u_2)$ has probability $\frac{1}{2}$, and the remaining output probability goes to vertex $x$, i.e. $(u_1,x)$ has probability $\frac{1}{2}$.
    Similarly, the edge $(v_1,v_2)$ has probability $\frac{1}{2}$, and the remaining probability is such that $(v_1,x)$ has probability $1/2$.
    \item  From layer $L_2$ to layer $L_3$ - the edge $(u_2,u_3)$ has probability $\frac{1}{2}$, and $(u_2,y)$ has the remaining probability $\frac{1}{2}$. Similarly, edge $(v_2,v_3)$ has probability $\frac{1}{2}$, and the remaining outgoing probability is such that $(v_2,y)$ has probability $\frac{1}{2}$. Edge $(x,y)$ has probability $1$.
    \item  From layer $L_3$ to layer $L_4$ - the edge $(u_3,u_4)$ has probability $\frac{1}{2}$ and the edge $(u_3,z)$ has the remaining outgoing probability $\frac{1}{2}$ from vertex $u_3$. Similarly, the edge $(v_3,u_4)$ has probability $\frac{1}{2}$ and the edge $(v_3,z)$ has the remaining outgoing probability $\frac{1}{2}$ from vertex $v_3$. Edge $(y,z)$ has probability $1$.
\end{enumerate}

We refer to the path $(u_1,u_2,u_3,u_4)$ as path $P_1$ or the upper path and path $(v_1,v_2,v_3,u_4)$ as path $P_2$ or the lower path.
Our proof relies on the following two lemmas:
\begin{lemma}
The ex-ante maximin value that can be guaranteed for instance $\cI$ is at least 
\[
U \triangleq \frac{1}{2} \left(\frac{1}{8}+\left(\frac{1}{2}+\frac{B}{6}\right)^3\right).
\]
\end{lemma}

\begin{proof}
We show a candidate solution that gives expected reward $U$ to both agents. We assign a budget of $B/3$ each for the transition between $L_1$ and $L_2$, the transition between $L_2$ and $L_3$, and the transition between $L_3$ and $L_4$. We construct a feasible transition matrix $M_1$ by decreasing the budget invested between any two layers in $\cI_2$ by $B/6$ (which can be done for $B$ small enough), and increasing the probabilities of edges $(u_1,u_2)$, $(u_2,u_3)$, and $(u_3,u_4)$ by $B/6$. We leave other edges in $\cI_1$ untouched. $M_1$ yields reward $\left(\frac{1}{2} + \frac{B}{6}\right)^3$ for $u_1$ and $1/8$ for $u_2$. Symmetrically, we construct a solution $M_2$ with expected reward $1/8$ for $u_1$ and $\left(\frac{1}{2} + \frac{B}{6}\right)^3$ for $u_2$. The solution that picks $M_1$ with probability $1/2$ and $M_2$ with probability $1/2$ guarantees maximin value of $U$.
\end{proof}

We now show that the ex-post maximin value is strictly smaller than $U$. We will reason exclusively on $\cI_1$, noting that when the budget $B$ is small and $w$ is large enough, investing any money on edges not contained within $\cI_1$ is sub-optimal. This can be seen immediately: %as $B$ goes to $0$, both the input probability and the reward of any node $v \in \cI_2$ go to $0$ for any feasible allocation, 
since one cannot invest in outgoing edges from nodes $x$ and $y$ in $\cI_2$, these nodes point to $z$ with probability $1$ (as $(x,y)$ and $(y,z)$ are not malleable) and have reward $0$; on the other hand, the input probability and the reward of any vertex in $\cI_1$ is strictly positive for $B$ small enough. In turn, for $B$ small enough, nodes in $\cI_1$ always have strictly higher rewards and input probabilities than any node in $\cI_2$ for any feasible allocation, and it is optimal to invest in improving only the transition probabilities from nodes in $\cI_1$ to nodes in $\cI_1$. 

To reason about how to optimally use the budget to improve edges in $\cI_1$, we introduce new notations. We let $W(v)$ denote the expected reward obtained by an agent starting at any vertex $v$ in the instance. We define $w^i_u \triangleq W(u_i)$ and $w^i_v \triangleq W(v_i)$ for $i \in \{1,2,3,4\}$. We let $B_1$, $B_2$, $B_3$ denote the budget split across the three transitions, and let $B_{\ge i}$ represent the budget spent to the right of layer $L_i$. Note that, in our instance, $B_{\ge 3} = B_3$. We now state the following the key lemma and show how it implies the desired separation.
\begin{lemma}\label{lemma:sumwelfarebound}
\begin{align*}
w^1_u + w^1_v &\le 2U.
\end{align*} 
Further, $w^1_u + w^1_v = 2U$ only holds when the budget is split equally across the transition between layers, i.e., $B_1 = B_2 = B_3 = \frac{B}{3}$; in this case,
\begin{align*}
w^2_u + w^2_v \leq \frac{1}{4} + \left(\frac{1}{2} + \frac{B}{6}\right)^2,
\end{align*}
\begin{align*}
w^2_u, w^2_v \leq \left(\frac{1}{2} + \frac{B}{6}\right)^2.
\end{align*}
% and $W(x_i) = 0$ for all $i$.
% \juba{and here we need to say that also, we have not invested in any leakage edge. Added.}
\end{lemma}

We defer the proof of Lemma~\ref{lemma:sumwelfarebound} to Appendix~\ref{app:sumwelfarebound}, and conclude the proof of separation by showing that Lemma~\ref{lemma:sumwelfarebound} implies the following corollary:
\begin{corollary}
The ex-post maximin value is strictly less than $U$.
\end{corollary}

\begin{proof}
Observe that the minimum reward over all agents, $\min\{w^1_u,w^1_v\}$, is upper bounded by $\frac{w^1_u + w^1_v}{2}$. Thus, by Lemma~\ref{lemma:sumwelfarebound}, we have that $\min\{w^1_u,w^1_v\} \le \frac{w^1_u + w^1_v}{2} \le U$. If $w^1_u + w^1_v < 2U$, the result holds. Hence, we only need to consider the case when $w^1_u + w^1_v = 2 U$. In this case, the budget split is given by $B_1 = B_2 = B_3 = B/3$ by Lemma~\ref{lemma:sumwelfarebound}. We assume with loss of generality that $w^2_u \geq w^2_v$ (otherwise invert the roles of $u_2$ and $v_2$); note that when the budget is small enough, it must be the case that $w^2_u,w^2_v$ are necessarily bigger than the rewards of node $x$ in layer $2$ in $\cI_2$ --- as $B$ tends to $0$, the rewards on nodes in the second layer in $\cI_1$ tend to $1/4$, while the reward of node $x$ remains $0$ always by non-malleability of $(x,y)$ and $(y,z)$). This directly implies that there exists an optimal maximin solution in which all of the budget $B/6$ allocated to improving edges (remembering that we need to spend half the budget, i.e. $B/6$, decreasing edges in $\cI_2$ for our matrix to remain stochastic) is allocated to $(u_1,u_2)$ and $(v_1,u_2)$. Let $\delta \in [0,1]$ be such that the probability transition of edge $(u_1,u_2)$ is set to $\frac{1}{2} + \delta \frac{B}{6}$, and the probability transition for $(v_1,u_2)$ is set to $(1-\delta) \frac{B}{6}$. We have that there exists $\delta$ such that
\begin{align*}
w_u^1 + w_v^1 
&\leq \left(\frac{1}{2} + \delta \frac{B}{6}\right) w_u^2 + (1-\delta) \frac{B}{6} w_u^2 + \frac{1}{2} w_v^2
\\&= \left(\frac{1}{2} + \frac{B}{6}\right) w_u^2 + \frac{1}{2} w_v^2,
\end{align*}
noting that all other transitions must go to $x$ and yield reward $0$, as with probability $1$, $x$ goes to $y$ and $y$ goes to $z$ with reward $0$ by non-malleability. Since $w_u^2 \leq \left(\frac{1}{2} + \frac{B}{6}\right)^2$, it must be that $w_u^2 = \left(\frac{1}{2} + \frac{B}{6}\right)^2$ and $w_v^2 = 1/4$, otherwise we would have $w_u^1 + w_v^1 < 2U$ which is a contradiction. This implies in particular that 
\begin{align*}
w_v^1 
&\leq \frac{B}{6} w_u^2 + \frac{1}{2} w_v^2
\\& = \frac{B}{6} \left(\frac{1}{2} + \frac{B}{6}\right)^2 + \frac{1}{8}.
\end{align*}
A simple calculation shows that for $B$ small enough, 
\[
\frac{B}{6} \left(\frac{1}{2} + \frac{B}{6}\right)^2 + \frac{1}{8} 
< \frac{1}{2} \left(\frac{1}{2} + \frac{B}{6}\right)^3 + \frac{1}{16} = U,
\]
which concludes the proof.
\end{proof}

\subsection{Proof of Lemma~\ref{lemma:sumwelfarebound}}\label{app:sumwelfarebound}
We prove the lemma by proving a similar result at each layer by induction, starting backward from the penultimate layer. 

\begin{lemma}
\label{lemma:generalsumwelfarebound}
For any layer $L_i$ with $i \in \{1,2,3\}$, we have 
\begin{align*}w^i_u + w^i_v &\le \left( \frac{1}{2^{4-i}} + \left(\frac{1}{2}+\frac{B_{\ge i}}{2(4-i)}\right)^{4-i} \right) 
\\ \max\{w^i_u,w^i_v\} &\le \left(\frac{1}{2}+\frac{B_{\ge i}}{2(4-i)}\right)^{4-i}
\end{align*} 
The first inequality is tight at layer $L_i$ only when $B_{\ge i}$ is split equally across the transition between layers to the right of $L_i$.%, and all vertices in $\mathcal{I}_2$ on layer $L_i$ yield reward $0$. 
\end{lemma}

We note that the above lemma applied at layer $1$ directly gives that $w_u^1 + w_v^1 \geq 2U$ only when $B_1 = B_2 = B_3 = \frac{B}{3}$ when equality holds, and all vertices in $\cI_2$ have $0$ reward. The second part of Lemma~\ref{lemma:sumwelfarebound} holds from applying Lemma~\ref{lemma:generalsumwelfarebound} at layer $2$ with $B_2 = B_3 = \frac{B}{3}$ or equivalently, $B_{\geq 2} = \frac{2B}{3}$.

\begin{proof}
Note that this lemma is trivially true for any layer $L_i$ when $B_{\ge i} = 0$. Henceforth, we only look at layer $L_i$ when $B_{\ge i} > 0$.

We begin by proving the lemma statement for Layer $L_3$ and work backwards towards layer $L_1$. Note that $w^3_u + w^3_v$ is maximized only by spending $B_{\ge 3} = B_3$ on edges from either $u_3$ or $v_3$ to the reward node $u_4$ with reward $1$. This gives us $w^3_u+w^3_v \le \frac{1}{2}+\frac{1}{2}+\frac{B_{\ge 3}}{2}$. Without loss of generality (due to the symmetry in the instance), let $w^3_u \ge w^3_v$. Then $w^3_u$ is maximized by spending all the budget $B_{\ge 3}$ on edge $(u_3,u_4)$, i.e., path $P_1$ (in the other case, all the budget is spent on path $P_2$). Thus, $\max\{w^3_u,w^3_v\} \le \frac{1}{2}+\frac{B_{\ge 3}}{2}$. Thus, both parts of the lemma are true for Layer $L_3$.

Now, let us assume our induction hypothesis holds at layer $L_{i+1}$. Consider layer $L_i$. Let us assume w.l.o.g that $w^{i+1}_u \ge w^{i+1}_v$. To maximize $w^i_u+w^i_v$, it is easy to see that all budget must be spent on edges from $u_i$ or $v_i$ to $u_{i+1}$ --- since $x$, $y$, and $z$ always have reward $0$. Thus, we get:
\begin{align}\label{eq:bound_induction}
\begin{split}
    w^i_u + w^i_v
    &\le \left(\frac{1}{2}+\frac{B_i}{2}\right) w^{i+1}_u + \frac{1}{2} w^{i+1}_v \\
    &\le \left(\frac{1}{2}+\frac{B_i}{2}\right) w^{i+1}_u 
    + \frac{1}{2} 
    \left(\frac{1}{2^{4-i-1}} + \left(\frac{1}{2} + \frac{B_{\geq i+1}}{2(4-i-1)}\right)^{4-i-1} - w_u^{i+1}\right)  \\
    &= \frac{B_i}{2}w_u^{i+1} + \frac{1}{2^{4-i}}  
    + \frac{1}{2} \left(\frac{1}{2} + \frac{B_{\geq i+1}}{2(4-i-1)}\right)^{4-i-1}\\
    &\leq \frac{1}{2^{4-i}}   
    +\frac{B_i}{2} \left(\frac{1}{2} + \frac{B_{\geq i+1}}{2(4-i-1)}\right)^{4-i-1}
    + \frac{1}{2} \left(\frac{1}{2} + \frac{B_{\geq i+1}}{2(4-i-1)}\right)^{4-i-1}\\
    &= \frac{1}{2^{4-i}}   
    +\left(\frac{1}{2} + \frac{B_i}{2}\right) \left(\frac{1}{2} + \frac{B_{\geq i + 1}}{2(4-i-1)}\right)^{4-i-1},
    \end{split}
\end{align}
where the second and second-to-last inequalities follow from our induction hypothesis. When $i = 2$, the above bound becomes 
\[
\frac{1}{4} + \left(\frac{1}{2} + \frac{B_2}{2}\right) \left(\frac{1}{2} + \frac{B_3}{2}\right),
\]
which is uniquely maximized (given total budget $B_{\geq 2}$ for layers more than $2$) if and only if $B_2 = B_3 = B_{\geq 3}$. When $i = 1$, this bound becomes 
\[
\frac{1}{8} + \left(\frac{1}{2} + \frac{B_1}{2}\right) \left(\frac{1}{2} + \frac{B_2 + B_3}{4}\right)^2,
\]
which is similarly uniquely maximized (when $B = B_1 + B_2 + B_3$) if and only if $B_1 = \frac{B_2 + B_3}{2}$, i.e. only if $B_1 = B/3$, $B_2 + B_3 = 2B/3$. In both cases, the unique maximizer satisfies $B_i = \frac{B_{\geq i}}{4-i}$ and $\frac{B_{\geq i+1}}{4-i-1} = B_i$. Therefore,
\begin{align*}
    w^i_u + w^i_v
    &\leq \frac{1}{2^{4-i}}   
    +\left(\frac{1}{2} + \frac{B_{\geq i}}{2(4-i)}\right) \left(\frac{1}{2} + \frac{B_{\geq i}}{2(4-i)}\right)^{4-i-1}
    \\&= \frac{1}{2^{4-i}}   
    +\left(\frac{1}{2} + \frac{B_{\geq i}}{2(4-i)}\right)^{4-i},
\end{align*}
and this equality can only be tight when i) $B_i = \frac{B_{\geq i}}{4-i}$ (by the unique maximizer argument above) and ii) the second inequality in Equation~\eqref{eq:bound_induction} is tight, which means 
\[
B_{\geq i+1} = B_{\geq i} - B_i = \frac{4-i-1}{4-i} B_{\geq i}
\]
is split equally across the $4-i-1$ transitions between layers to the right of $L_{i+1}$, by induction hypothesis. This immediately implies that $B_i = \ldots = B_3 = \frac{B_{\geq i}}{4-i}$ when the inequality is tight.

We conclude our proof by showing an upper bound on $w_u^i,w_v^i$. By induction hypothesis, 
\[
w_u^{i+1},w_v^{i+1} \leq \left(\frac{1}{2}+\frac{B_{\ge i+1}}{2(4-i-1)}\right)^{4-i-1}.
\]
It immediately implies that 
\[
w_u^{i+},w_v^{i} \leq \left(\frac{1}{2} + \frac{B_{i}}{2}\right) \left(\frac{1}{2}+\frac{B_{\ge i+1}}{2(4-i-1)}\right)^{4-i-1},
\]
noting that the maximum transition probability from either $w_u^i$ or $w_v^i$ to the best of $w_u^{i+1},~w_v^{i+1}$ is at most $\frac{1}{2} + \frac{B_i}{2}$ (half the budget must be spent decreasing other edges, and half the budget $B_i$ is spent increasing transitions to the best node in $\cI_1$ on layer $i+1$). By the exact same argument as for the first part of the lemma, this is upper-bounded by $\left(\frac{1}{2} + \frac{B_{\geq i}}{2(4-i)}\right)^{4-i}$. Hence the induction hypothesis holds at layer $i$.
\end{proof}

\section{Omitted Proofs}
\subsection{Proof of Theorem~\ref{thm:MMW_approximation}: Algorithmic Guarantees for Ex-Post-Fairness}\label{app:maximin_program}

Recall that $OPT_{MM}^\varepsilon$ is the optimum maximin value under discretized splits of the budget. Let $M_1^\varepsilon, \ldots, M_{k-1}^\varepsilon$ be a set of transition matrices with expected reward for each starting position $i$ lower-bounded by by $R^\top M_{k-1}^\varepsilon \ldots M_1^\varepsilon e_i \geq OPT_{MM}^\varepsilon \triangleq OPT_{MM} - (k-1) \varepsilon \left\Vert R \right\Vert_{\infty}$ that is feasible with respect to budget split $B_1^\varepsilon,\ldots, B_{k-1}^\varepsilon$. Note that such matrices exist by Claim~\ref{clm: budget_loss}. Let $E_t \in \cD^w$ denote the population-wise probability distribution that is induced by these transition matrices on layer $t$, i.e.
\[
E_t^j = M^\varepsilon_{t-1} \ldots M^\varepsilon_1 e_j~\forall j \in [w].
\]

To prove the result, we will show by induction that for all $B_{\geq t} \geq B^\varepsilon_{\geq t}$, for $A_t \in \cA(\varepsilon)$ such that $\Vert A^j_t - E^j_t \Vert_1 \leq \varepsilon~~\forall j \in [w]$, we have
\[
R^\top M(B_{\geq t},A_t) A^j_t \geq OPT_{MM}^\varepsilon - 2(k - t) \varepsilon \Vert  R \Vert_{\infty}, ~~\forall j \in [w],
\]
i.e., a population-wise welfare approximation guarantee starting at any layer $t$. Since we can take $B_{\geq 1} = B^\varepsilon$, this directly implies
\[
R^\top M(B^\varepsilon,A_1) e_j \geq OPT_{MM}^\varepsilon - 2(k-1) \varepsilon \left\Vert R \right\Vert_{\infty}~\forall j \in [w].
\]
Combined with Claim~\ref{clm: budget_loss}, which states that $OPT_{MM}^\varepsilon \geq OPT_{MM} - (k-1) \varepsilon \left\Vert R \right\Vert_{\infty}$, we obtain the result. 

Let us now provide our inductive proof. First, consider the transition from layer $L_{k-1}$ to layer $L_k$. Note that 
\[
OPT_{MM}^\varepsilon \le  R^\top M_{k-1}^\varepsilon\ldots M_1^\varepsilon e_j = R^\top M_{k-1}^\varepsilon E_{k-1}^j~~\forall j \in [w]
\]
using the fact that $E_{t+1}^j = M_t^\varepsilon E_t^j$. Let $A_{k-1} \in \cA(\varepsilon)$ be such that $\Vert A^j_{k-1} - E^j_{k-1} \Vert \leq \varepsilon~~\forall j \in [w]$. Note that such a $A_{k-1}$ always exists (by definition of $\cA(\varepsilon)$), and is considered by Algorithm~\ref{alg: max_MM}. By Corollary~\ref{cor:approx_loss}, 
\[
R^\top M_{k-1}^\varepsilon A^j_{k-1} \geq R^\top M_{k-1}^\varepsilon E^j_{k-1} - \varepsilon \Vert  R \Vert_{\infty} \geq OPT_{MM}^\varepsilon  - \varepsilon \Vert  R \Vert_{\infty}~~\forall j \in [w].
\]
Further, $M_{k-1}^\varepsilon$ is feasible for Program~\eqref{MMW_DP} with respect to $B_{\geq k-1}, B_{\geq k} = 0$, given $B_{\geq k-1} \geq B_{\geq k-1}^\varepsilon$. As such, for $B_{\geq k-1} \geq B_{\geq k-1}^\varepsilon$, by optimality of $M(B_{\geq k-1},A_{k-1})$, we have that
\[
\min_{j \in [w]} R^\top M(B_{\geq k-1},A_{k-1}) A^j_{k-1}  \geq \min_{j \in [w]} R^\top M_{k-1}^\varepsilon A^j_{k-1},
\]
and in turn
\[
R^\top M(B_{\geq k-1},A_{k-1}) A^j_{k-1} 
\geq OPT_{MM}^\varepsilon - \varepsilon \Vert  R \Vert_{\infty} ~~\forall j \in [w].
\]

Now, suppose the induction hypothesis holds at layer $t+1$. I.e., for all $B_{\geq t+1} \geq B^\varepsilon_{\geq t+1}$, for $A_{t+1} \in \cA(\varepsilon)$ such that $\Vert A^j_{t+1} - E_{t+1}^j \Vert_1 \leq \varepsilon~~\forall j \in [w]$,
\[
R^\top M(B_{\geq t+1},A_{t+1}) A^j_{t+1} \geq OPT_{MM}^\varepsilon - 2(k - t - 1) \varepsilon \Vert  R \Vert_{\infty}~~\forall j \in [w].
\]
For any given $B_{\geq t} \geq B_{\geq t}^\varepsilon$, note that one can set $B_{\geq t+1} = B_{\geq t+1}^\varepsilon$ and have $B_t \geq B_t^\varepsilon$; hence, $M_t^\varepsilon$ is feasible for Program~\eqref{MMW_DP} with respect to $B_t \geq B_t^\varepsilon,B^\varepsilon_{\geq t+1}$. Consider $A_t \in \cA(\varepsilon)$ such that $\Vert A^j_{t} - E_{t}^j \Vert_1 \leq \varepsilon~~\forall j \in [w]$. Note that such a $A_{t}$ always exists (by definition of $\cA(\varepsilon)$, and is considered by Algorithm~\ref{alg: max_MM}. Since we have $\Vert A^j_t -E_t^J \Vert_1 \leq \varepsilon$ and $\Vert A^j_{t+1} - M_t^\varepsilon E_t^j \Vert_1 \leq \varepsilon$ $\forall j \in [w]$, applying Corollary~\ref{cor:approx_loss} yields that $\forall j \in [w]$,  
\begin{align*}
R^\top M(B^\varepsilon_{\geq t+1}, A_{t+1}) M^\varepsilon_t A^j_t 
&\geq R^\top M(B^\varepsilon_{\geq t+1}, A_{t+1}) M^\varepsilon_t E_t^j - \varepsilon \Vert  R \Vert_{\infty}
\\ &\geq R^\top M(B^\varepsilon_{\geq t+1}, A_{t+1}) A^j_{t+1} - 2 \varepsilon \Vert  R \Vert_{\infty}.
\end{align*}
Using the induction hypothesis, we obtain that
\[
R^\top M(B^\varepsilon_{\geq t+1}, D_{t+1}) M^\varepsilon_t A^j_t \geq OPT_{MM}^\varepsilon - 2(k - t) \varepsilon \Vert  R \Vert_{\infty}~~\forall j \in [w],
\]
which can be rewritten as 
\[
\min_{j \in [w]} R^\top M(B^\varepsilon_{\geq t+1}, D_{t+1}) M^\varepsilon_t A^j_t \geq OPT_{MM}^\varepsilon - 2(k - t) \varepsilon \Vert  R \Vert_{\infty}.
\]
In particular, by optimality of $M(B_{\geq t},A_t)$, it must be the case that
\[
\min_{j \in [w]} R^\top M(B_{\geq t}, A_{t}) A^j_t \geq OPT_{MM}^\varepsilon - 2(k - t) \varepsilon \Vert  R \Vert_{\infty},
\]
which concludes the proof of the accuracy guarantee. The running time is obtained noting that at each time step $t$, we solve one Program~\ref{MMW_DP} for each of the (at most) $\frac{B}{\varepsilon}$ possible budget splits of $B_{\geq t}$ and for each of the $\left(\left(\frac{1}{\varepsilon}\right)^w\right)^w =
\left(\frac{1}{\varepsilon}\right)^{w^2}$ population-wise probability distributions in $\cA(\varepsilon)$ on both layer $L_t$ and layer $L_{t+1}$; i.e., in a given time step, the algorithm solves $O\left( \frac{B}{\varepsilon} \left(\frac{1}{\varepsilon}\right)^{w^4}\right)$ optimization programs. Then, the algorithm finds the solution of all of these programs with the best objective value, which can be done in time linear in the number of such solutions, i.e. $O \left(\frac{B}{\varepsilon} \left(\frac{k}{\varepsilon}\right)^{w^4}\right)$. The algorithm does so over $k-1$ time steps.

\subsection{Proof of Lemma~\ref{lem:exante_accuracy}: Algorithmic Guarantees for Ex-Ante Fairness}\label{app:exante_accuracy}

The proof follows that of Theorem 1 of \citet{FS96}. Note that we can rewrite the objective in the normal form given in Corollary 2 of \citet{FS96}, by letting the payoff matrix $G$ be such that $G\left(\left(M_1,\ldots,M_{k-1}\right),q\right) = R^\top M_{k-1} \ldots M_1 e_q$ when the designer plays $\left(M_1,\ldots,M_{k-1}\right) \in \cF$ and the learner plays $q \in [w]$. Noting that $G$ has entries bounded between $0$ and $\Vert R \Vert_{\infty}$, we can apply Corollary 2 of \citet{FS96} to loss $R^\top M^t_{k-1} \ldots M^t_1 D^t$ with an appropriate renormalization to show the following low-regret statement:
\[
\frac{1}{T} \sum_{t=1}^T R^\top M^t_{k-1} \ldots M^t_1 D^t \leq \min_{D \in \cD} \frac{1}{T} \sum_{t=1}^T R^\top M^t_{k-1} \ldots M^t_1 D 
+\left( \sqrt{2 \frac{\ln w}{T}} + \frac{\ln w}{T} \right) \Vert R \Vert_{\infty},
\]
Since $R^\top M^t_{k-1} \ldots M^t_1 D$ is linear in $D$, we have
\[
\min_{D \in \cD} \frac{1}{T} \sum_{t=1}^T R^\top M^t_{k-1} \ldots M^t_1 D = \min_{q \in [w]} \frac{1}{T} \sum_{t=1}^T R^\top M^t_{k-1} \ldots M^t_1 e_q,
\]
and the following low-regret statement
\begin{align}\label{eq: regret_MW}
\frac{1}{T} \sum_{t=1}^T R^\top M^t_{k-1} \ldots M^t_1 D^t 
\leq 
\min_{q \in [w]} \frac{1}{T} \sum_{t=1}^T R^\top M^t_{k-1} \ldots M^t_1 e_q 
+ \left( \sqrt{2 \frac{\ln w}{T}} + \frac{\ln w}{T} \right) \Vert R \Vert_{\infty}.
\end{align}

We can now show the result, using a similar argument to that of Theorem 1 of \citet{FS96}. To do so, we let $\bar{D} \in \cD$ be the probability distribution given by $\bar{D} \triangleq \frac{1}{T} \sum_{t=1}^T D^t$. We have that
\begin{align*}
    &\min_{q \in [w]} R^\top \E_{M \sim \overline{\Delta M}} \E \left[M_{k-1} \ldots M_1\right] e_q
    \\&= \min_{q \in [w]} \frac{1}{T} \sum_{t=1}^T R^\top M^t_{k-1} \ldots M^t_1 e_q \tag{by definition of $\overline{\Delta M}$}
    \\&\geq \frac{1}{T} \sum_{t=1}^T R^\top M^t_{k-1} \ldots M^t_1 D^t 
    - \left( \sqrt{2 \frac{\ln w}{T}} + \frac{\ln w}{T} \right) \Vert R \Vert_{\infty}
    \tag{by Equation~\eqref{eq: regret_MW}}
    \\&\geq  \frac{1}{T} \sum_{t=1}^T \max_{M \in \cF} R^\top M_{k-1} \ldots M_1 D^t
    -\varepsilon
    -\left( \sqrt{2 \frac{\ln w}{T}} + \frac{\ln w}{T} \right) \Vert R \Vert_{\infty} \tag{$M^t$ $\varepsilon$-approx. best response to $D^t$}
    \\&\geq \max_{M \in \cF} \frac{1}{T} \sum_{t=1}^T R^\top M_{k-1} \ldots M_1 D^t
    -\varepsilon
    -\left( \sqrt{2 \frac{\ln w}{T}} + \frac{\ln w}{T} \right) \Vert R \Vert_{\infty}
    \tag{Max of sum less than sum of max}
    \\& \geq \max_{\Delta M \in \Delta \cF} \frac{1}{T} \sum_{t=1}^T  \E_{M \sim \Delta M} \left[R^\top M_{k-1} \ldots M_1 D^t\right] 
    -\varepsilon \tag{Expectation over $\Delta M$ less than best realization of $\Delta M$, and the realization is in $\cF$}
    -\left( \sqrt{2 \frac{\ln w}{T}} + \frac{\ln w}{T} \right) \Vert R \Vert_{\infty}
    \\& = \max_{\Delta M \in \Delta \cF} R^\top \E_{M \sim \Delta M} \left[M_{k-1} \ldots M_1\right] \bar{D}
    -\varepsilon
    -\left( \sqrt{2 \frac{\ln w}{T}} + \frac{\ln w}{T} \right) \Vert R \Vert_{\infty}
    \\&\geq \min_{D \in \cD}  \max_{M \in \Delta\cF} R^\top \E_{M \sim \Delta M} \left[M_{k-1} \ldots M_1\right] D
    -\varepsilon
    -\left( \sqrt{2 \frac{\ln w}{T}} + \frac{\ln w}{T} \right) \Vert R \Vert_{\infty}
    \\&\geq \max_{\Delta M \in \Delta \cF} \min_{D \in \cD} R^\top \E_{M \sim \Delta M} \left[M_{k-1} \ldots M_1\right] D
    -\varepsilon
    -\left( \sqrt{2 \frac{\ln w}{T}} + \frac{\ln w}{T} \right) \Vert R \Vert_{\infty} \tag{Max-min inequality}
    \\&= \max_{\Delta M \in \Delta \cF} \min_{q \in [w]} R^\top \E_{M \sim \Delta M} \left[M_{k-1} \ldots M_1\right] e_q
    -\varepsilon
    -\left( \sqrt{2 \frac{\ln w}{T}} + \frac{\ln w}{T} \right) \Vert R \Vert_{\infty}.
\end{align*}
This concludes the proof.

\subsection{Omitted Proofs for Price of Fairness}\label{app:price_fairness}

\subsubsection{Proof of the Lower Bound of Theorem~\ref{thm: fairness_lb}}\label{app:fairness_lb}

Note that $P_f \geq 1$ is always true, by definition. The proof of the other two cases when $B \leq 2w$ is based on Example~\ref{ex: lb_construction}. We divide the analysis of the construction into the following cases:
\begin{enumerate}
\item $B \leq 2$. It is easy to see that $OPT_{SW} = \frac{B}{2} (1-(w-1) \varepsilon)$, and is achieved by the following transition matrix:
\begin{align*}
M^*_1 = \left(
\begin{matrix}
B/2 & 0 & \ldots & 0
\\1 - B/2 & 1 & \ldots & 1
\end{matrix}
\right)
\end{align*}
Now note that the maximin solution is unique (and in particular, is the maximim solution with the highest social welfare), and this unique maximin solution splits the budget evenly among the starting nodes and yields social welfare $\frac{B}{2w}$, via transition matrix
\begin{align*}
M_1^{f} = \left(
\begin{matrix}
\frac{B}{2w} & \ldots & \frac{B}{2w} 
\\1- \frac{B}{2w} & \ldots & 1- \frac{B}{2w}
\end{matrix}
\right)
\end{align*}
Therefore, we have that 
\[
P_f^+(\varepsilon) = \frac{B(1- w\varepsilon)/2}{\frac{B}{2w}} = w(1 - w \varepsilon),
\]
and
\[
\lim_{\varepsilon \to 0} P_f(\varepsilon) = w.
\]

\item Now, consider the case when $2 \leq B \leq 2w$. On the one hand, note that $OPT_{SW} \geq 1 - (w-1)\varepsilon$, as setting
\begin{align*}
M^*_1 = \left(
\begin{matrix}
1 & 0 & \ldots & 0
\\ 0 & 1 & \ldots & 1
\end{matrix}
\right)
\end{align*}
only requires a budget of $2$ hence is feasible for Program~\eqref{SW_program}. The unique maximin solution is still given by $M_1^f$ and has welfare $\frac{B}{2w}$. As such, we have 
\[
P_f(\varepsilon) \geq \frac{1- w \varepsilon}{\frac{B}{2w}} = 2w \frac{1 - (w-1)\varepsilon}{B}.
\]
In particular, taking $\varepsilon \to 0$, we get that a lower bound on the price of fairness is given by $\frac{2w}{B}$.
\end{enumerate}
The proof for $B \geq 2w$ is immediate, noting that 
\begin{align*}
M^*_1 = \left(
\begin{matrix}
1 & 1 & \ldots & 1
\\ 0 & 0 & \ldots & 0
\end{matrix}
\right)
\end{align*}
is feasible. As such $OPT_{SW} = 1$, and $M^*_1$ is a maximin solution with welfare $1$.

\subsubsection{Proof of Lemma~\ref{lem: ub_OPT}}\label{app:ub_OPT}

\begin{proof}
The first part of the claim is immediate from noting that given an optimal solution $M_1^*,\ldots,M_{k-1}^*$ to Program~\eqref{SW_program}, 
\[
OPT_{SW} = R^\top M_{k-1}^* \ldots M_1^* D_1^0
\leq \left\Vert R \right\Vert_{\infty} \left\Vert M_{k-1}^* \ldots M_1^* D_1^0 \right\Vert
= \left\Vert R \right\Vert_{\infty},
\]
where the last equality follows from $M_{k-1}^* \ldots M_1^* D_1^0$ being a probability distribution.

For the second part of the claim, consider any feasible solution $M_1,\ldots,M_{k-1}$ with corresponding split $B_1, \ldots, B_{k-1}$ of the budget. I.e., $B = \sum_{t = 1}^{k-1} B_t$, and $\sum_i \sum_j \left\vert M_t(i,j) - M_t^0(i,j) \right\vert \leq B$ for all $t$. Note that at layer $L_{t}$, for any input distribution $D_{t}$, and vector $R_{t+1}$ with non-negative coordinates at layer $t+1$, we have that 
\begin{align*}
R_{t+1}^\top \left(M_{t} - M_{t}^0\right) D_{t}
&= \sum_{i=1}^w R_{t+1}(i) \sum_{j=1}^w \left(M_t(i,j) - M_t^0(i,j)\right)  D_t(j)
\\&= \sum_{j=1}^w D_t(j) \sum_{i=1}^w R_{t+1}(i) \left(M_t(i,j) - M_t^0(i,j)\right)  
\\&\leq \sum_{j=1}^w D_t(j) 
\sum_{i \in S_j} R_{t+1}(i) \left(M_t(i,j) - M_t^0(i,j)\right)  
\\&\leq \left\Vert R_{t+1} \right\Vert_{\infty }
\sum_{j=1}^w D_t(j) \sum_{i \in S_j} M_t(i,j) - M_t^0(i,j)
\end{align*}
where $S_j = \{i:~ M_t(i,j) - M_t^0(i,j) \geq 0\}$ and where the second-to-last inequality follows from the fact that $R_{t+1}(i) \geq 0$. As $M_t,M_t^0 \in \cM$, we have that
\[
\sum_{i=1}^w  M_t(i,j) - M_t^0(i,j) = 1 - 1 = 0,
\]
which implies that 
\begin{align*}
    \sum_{i \in S_j}  M_t(i,j) - M_t^0(i,j)
    = \sum_{i = 1}^w  M_t(i,j) - M_t^0(i,j) - \sum_{i \notin S_j} M_t(i,j) - M_t^0(i,j)
    = - \sum_{i \notin S_j} M_t(i,j) - M_t^0(i,j).
\end{align*}
In turn, we have that 
\begin{align*}
    \frac{B_t}{2} 
    \geq \frac{c(M_t,M_t^0)}{2}
    &=\frac{1}{2}\sum_j \sum_i \left\vert M_t(i,j) - M_t^0(i,j) \right\vert 
    \\&= \frac{1}{2} \sum_j \left(\sum_{i \in S_j} M_t(i,j) - M_t^0(i,j) -  \sum_{i \notin S_j} M_t(i,j) - M_t^0(i,j) \right)
    \\&= \sum_j \sum_{i \in S_j} M_t(i,j) - M_t^0(i,j).
\end{align*}
This can also be seen noting that increasing edges of $M_t^0$ by a total amount of $\delta B_t$ requires decreasing other edges by a total amount of $\delta B_t$ for $M_t$ to be a stochastic matrix and so requires a total budget of $2\delta B_t$, which in turn implies that $\delta \leq \frac{1}{2}$ necessarily. This implies that
\begin{align*}
R_{t+1}^\top \left(M_{t} - M_{t}^0\right) D_{t}
&\leq \left\Vert R_{t+1} \right\Vert_{\infty }
\sum_{j=1}^w D_t(j) \sum_{i \in S_j} M_t(i,j) - M_t^0(i,j)
\\&\leq \left\Vert R_{t+1} \right\Vert_{\infty }
\sum_{j=1}^w \sum_{i \in S_j} M_t(i,j) - M_t^0(i,j)
\\&\leq \frac{B_t}{2} \left\Vert R_{t+1} \right\Vert_{\infty }
\end{align*}
Applying the above inequality recursively, we have that 
\begin{align*}
    R^\top M_{k-1} \ldots M_{1} D_1^0 
    &\leq \frac{B_{k-1}}{2} \left\Vert R \right\Vert_{\infty} 
    + R^\top M^0_{k-1} M_{k-2} \ldots M_{1} D_1^0
    \\&\leq \frac{B_{k-1} + B_{k-2}}{2} \left\Vert R \right\Vert_{\infty}
    +  R^\top M^0_{k-1} M^0_{k-2} M_{k-3} \ldots M_{1} D_1^0
    \\&~~~~~~~~~~~~~~~~~~~~~~~~~~~\vdots
    \\&\leq \frac{\sum_t B_{t}}{2} \left\Vert R \right\Vert_{\infty}
    +  R^\top M^0_{k-1} \ldots M^0_{1} D_1^0
    \\& = \frac{B}{2} + W^0.
\end{align*}
\end{proof}

\subsubsection{Proof of Lemma~\ref{lem: lb_maxmin}}\label{app:lb_maxmin}

First, we consider the case when $\frac{B}{2w} \leq 1$. We focus on the transition from the second-to-last layer $L_{k-1}$ to the last layer $L_k$. Remember that on layer $L_k$, $R(1) =\Vert R \Vert_{\infty}$. We re-number (w.l.o.g.) the nodes on layer $k-1$ so that $M^0_{k-1}(1,i) \geq \frac{B}{2w}$ for all nodes $i \in [l]$, and $M^0_{k-1}(1,i) < \frac{B}{2w}$ for all nodes $i \in \{l+1,\ldots,k\}$, for some $l \in \{0,\ldots,w\}$. Let us set
\begin{align*}
M_{k-1} = M_{k-1}^0 +  \left(
\begin{matrix}
0 & 0 & \ldots & 0 & \frac{B}{2w} - M^0_{k-1}(1,l+1) & \ldots & \frac{B}{2w} - M^0_{k-1}(1,w)
\\0 & 0 & \ldots & 0 & -\alpha_{2,l+1} & \ldots & -\alpha_{2,w}& 
\\ \vdots & \vdots & \vdots & 0 & \vdots & \cdots & \vdots
\\0 & 0 & \ldots &  0 & -\alpha_{w,l+1} & \ldots & -\alpha_{w,w}& 
\end{matrix}
\right),
\end{align*} and
\begin{align*}
M_t = M_t^0~\forall t < k-1,
\end{align*}
where the $\alpha_{i,j}$'s are chosen to guarantee $M^0_{k-1}(i,j) \geq \alpha_{i,j} \geq 0$ for all $i,j \in [w]$ and $\sum_{i=2}^w \alpha_{i,j} = \frac{B}{2w} - M^0_{k-1}(1,j) > 0$ for all $j > l$. Such a choice of $\alpha_{i,j}$'s exists as 
\[
\sum_{j=2}^w M^0_{k-1}(i,j) = 1 - M^0_{k-1}(1,j) \geq /(2w) - M^0_{k-1}(1,j).
\]
Now, note that $M_{k-1} \in \cM$, as all coefficients are between $0$ and $1$ and the elements within the same column still sum to $1$ by choice of $\alpha_{i,j}$'s. Finally, 
\[
c(M_{k-1},M_{k-1}^0) 
= \sum_{j = l+1}^w 2 \left(\frac{B}{2w} - M^0_{k-1}(1,j)\right)
\leq \sum_{j = l+1}^w 2 \frac{B}{2w}
\leq B.
\]
Therefore, $(M_{k-1},\ldots,M_1)$ is a feasible solution for Programs~\ref{SW_program} and~\ref{maxmin_program} under budget $B$.

Now, note that by construction, $M_{k-1}(1,j) \geq \frac{B}{2w}$ for all $j \in [w]$. In turn, this implies that for any distribution $D_{k-1}$,
\begin{align*}
R^\top M_{k-1} D_{k-1} 
= \sum_{i=1}^w R(i) \left(M_{k-1} D_{k-1}\right)(i)
&= \sum_{i=1}^w R(i) \sum_{j=1}^w M_{k-1}(i,j) D_{k-1}(j)
\\& \geq R(1) \sum_{j=1}^w M_{k-1}(1,j) D_{k-1}(j)
\\&\geq R(1) \frac{B}{2w} \sum_{j=1}^w D_{k-1}(j)
\\&= \frac{B}{2w} \left\Vert R \right\Vert_{\infty}, 
\end{align*}
since $\sum_{j=1}^w D_{k-1}(j) = 1$ by virtue of $D_{k-1}$ being a probability distribution, and because $R(1) = \left\Vert R \right\Vert_{\infty}$ by choice of node indexing. In particular, for all $j$, $D_{k-1} = M_{k-2} \ldots M_1 e_j$ is a probability distribution, hence 
\[
R^\top M_{k-1} \ldots M_1 e_j \geq \frac{B}{2w} \left\Vert R \right\Vert_{\infty}.
\]
Therefore, there exists a solution to Program~\eqref{maxmin_program} that has value $\frac{B}{2w} \left\Vert R \right\Vert_{\infty}$, implying any optimal solution to Program~\eqref{maxmin_program} has value at least $\frac{B}{2w} \left\Vert R \right\Vert_{\infty}$. In turn, such an optimal solution $M^f_1, \ldots,M^f_{k-1} \in S^f$ must have welfare 
\begin{align*}
R^\top M^f_{k-1} \ldots M^f_1 D_1^0 
= \sum_{j=1}^w D_1^0(j) R^\top M^f_{k-1} \ldots M^f_1 e_j
\geq \sum_{j=1} D_1^0(j) \frac{B}{2w} \left\Vert R \right\Vert_{\infty} 
=\frac{B}{2w} \left\Vert R \right\Vert_{\infty},
\end{align*}
which concludes the proof when $B \leq 2w$. 

When $B > 2w$, let $B' = 2w$. By the above, any optimal solution to maximin Program~\eqref{maxmin_program} with budget $B' = 2w$ has value at least $\frac{B'}{2w} \left\Vert R \right\Vert_{\infty} = \left\Vert R \right\Vert_{\infty}$. This immediately implies that any optimal solution to Program~\eqref{maxmin_program} with budget $B$ also has value at least $\left\Vert R \right\Vert_{\infty}$, since any feasible solution for budget $B'$ is feasible for budget $B$. In turn, any optimal solution to maximin Program~\eqref{maxmin_program} under budget $B > 2w$ must have welfare at least $\left\Vert R \right\Vert_{\infty}$.

\subsubsection{Proof of Lemma~\ref{lem: lb_maxmin_w0}}\label{app:lb_maxmin_w0}

The proof of the lemma uses the following Claim~\ref{clm: maxmin_equalspending}, that shows that an optimal solution to Program~\eqref{maxmin_program} spends all the budget $B$:
\begin{claim}\label{clm: maxmin_equalspending}
Let $(M_1,\ldots,M_{k-1})$ be a solution to maximin Program~\eqref{maxmin_program} with social welfare strictly less than $\Vert R \Vert_{\infty}$. If all edges are malleable, it must be the case that $\sum_{t=1}^{k-1} c(M_t,M_t^0) = B$.
\end{claim}
The proof idea is simple: if $\sum_{t=1}^{k-1} c(M_t,M_t^0) < B$, the leftover budget can be used to improve the scial welfare, unless this social welfare already is the maximum achievable value of $\Vert R \Vert_{\infty}$.

\begin{proof}
By contradiction, suppose that $B > \sum_{t=1}^{k-1} c(M_t,M_t^0)$. Remember the numbering of nodes on layer $L_k$ is chosen such that $R(1) = \left\Vert R \right\Vert_{\infty}$. Let us pick all $j$ such that $M_{k-1}(1,j) < 1$ (if such a $j$ exists); note that since $M_{k-1}$ is stochastic, there also exists $q \neq 1$ such that $M_{k-1}(q,j) > 0$. For $\varepsilon$ arbitrarily small, there hence exists a matrix $M'_{k-1}(\varepsilon) \in \cM$ such that for all such $j$, $M'_{k-1}(1,j) = M_{k-1}(1,j) + \varepsilon$ and $\sum_{q \neq 1} M'_{k-1}(q,j) = \sum_{q \neq 1} M_{k-1}(q,j) - \varepsilon$, and such that $M'_{k-1}(q,j) = M_{k-1}(q,j)$ for all $q \in [w]$ and all $j$ with $M_{k-1}(1,j) = 1$.

Now, note that $c(M'_{k-1},M^0_{k-1})
\leq c(M'_{k-1},M_{k-1}) + c(M_{k-1},M^0_{k-1})$ with 
$
\lim_{\varepsilon \to 0} c(M'_{k-1},M_{k-1}) = 0
$. 
In turn, this implies that for $\varepsilon$ small enough,
$
\sum_{t=1}^{k-2} c(M_t,M_t^0) + c(M_t',M_t^0) \leq  B
$, 
hence $\left(M_1,\ldots,M_{k-2},M_{k-2}'\right)$ is feasible for Program~\eqref{maxmin_program}. Further, by construction, for all $j$ with $M_{k-1}(1,j) < 1$, we have
\begin{align*}
R^\top M'_{k-1} e_j - R^\top M_{k-1} e_j 
\geq \left(R(1) - \max_{q \neq 1} R(q)\right) \varepsilon
> 0,
\end{align*}
and for all $j$ with $M_{k-1}(1,j) = 1$, we have $R^\top M'_{k-1} e_j = R^\top M_{k-1} e_j$. Since for all starting nodes $i \in L_1$ such that $R^\top M_{k-1} \ldots M_{1} e_i < \left\Vert R \right\Vert_{\infty}$, there must exist $j$ such that $M_{k-1}(1,j) < 1$ and $\left(M_{k-2} \ldots M_1 \right)(j,i) > 0$ (otherwise $\left(M_{k-1} \ldots M_1 \right)(1,i) = 1$ and node $i$ obtains reward $\left\Vert R \right\Vert_{\infty}$), it immediately follows that for all such $i$,
\[
R^\top M'_{k-1} M_{k-2} \ldots M_1 e_i  > R^\top M_{k-1} M_{k-2} \ldots M_1 e_i.
\]
This contradicts $(M_1,\ldots,M_{k-1})$ being an optimal solution for Program~\eqref{maxmin_program}.
\end{proof}

We are now ready to prove Lemma~\ref{lem: lb_maxmin_w0}. Let $M_1^f,\ldots,M_{k-1}^f$ be an optimal solution for Program~\eqref{maxmin_program}. Note that if the maximin value of this solution is $\left\Vert R \right\Vert_{\infty}$, then the solution necessarily has welfare $\left\Vert R \right\Vert_{\infty} \geq W^0$, which concludes the proof. So, without loss of generality, we can assume there exists at least one starting node $q$ such that 
 \[
 R^\top M_{k-1}^f \ldots M_1^f e_q < \left\Vert R \right\Vert_{\infty}.
\]
Fix a layer $t$, and let $R_{out}^\top \triangleq R^\top M_k^f \ldots M_{t+1}^f$ and $M_{in} = M_{t-1}^f \ldots M_1^f$. Note that the utility obtained by the $q$-th node in the starting layer is immediately given by
\begin{align*}
R_i = \sum_{i=1}^w M_{in}(i,q) \sum_{j=1}^w M^f_t(j,i) R_{out}(j).
\end{align*}
Indeed, starting from node $q$ in the first layer, an individual transitions to node $i$ on layer $t$ with probability $M_{in}(i,q)$, then to node $j$ with reward $R_{out}(j)$ on layer $t+1$ with probability $M_t(j,i)$. 

Suppose by contradiction that for some $i'$, $\sum_{j=1}^w M^f_t(j,i') R_{out}(j) < \sum_{j=1}^w M^0_t(j,i') R_{out}(j)$ (necessarily, $M^f_t(j,i') \neq M^0_t(j,i')$ for some $j$). We will construct a set of transition matrices that achieves the same maximin value, but requires budget strictly less than $B$. To do so, let $M_t'$ be such that $M_t'(j,i') = M_t^0(j,i')$ for all $j$ and $M_t'(j,i) = M_t^f(j,i)$ for all $i \neq i'$, for all $j$. First, $c(M_t',M_t^0) < c(M_t^f,M_t^0)$, as
\begin{align*}
\sum_{i,j} \left\vert M_t'(j,i) - M_0(j,i)\right\vert
&= \sum_{i \neq i', j} \left\vert M_t^f(j,i) - M_0(j,i)\right\vert
\\&<  \sum_{i \neq i',j} \left\vert M_t^f(j,i) - M_0(j,i)\right\vert 
+\sum_{j} \left\vert M_t^f(j,i') - M_0(j,i')\right\vert
\\& = c(M_t,M_t^0)
\end{align*}
where the strict inequality follows from the fact that $M^f_t(j,i') \neq M^0_t(j,i')$ for some $j$. Second, for all $q$, we immediately have that as the $M_{in}(i,q)$ are non-negative, 
\[
\sum_{i=1}^w M_{in}(i,q) \sum_{j=1}^w M'_t(j,i) R_{out}(j) \geq \sum_{i=1}^w M_{in}(i,q) \sum_{q=1}^w M^f_t(j,i) R_{out}(j),
\]
since for all $i \neq i'$ we have $\sum_{j=1}^w M'_t(j,i) R_{out}(j) = \sum_{j=1}^w M^f_t(j,i) R_{out}(j)$, and by construction 
\\$\sum_{j=1}^w M'_t(j,i') R_{out}(j) = \sum_{j=1}^w M^0_t(j,i') R_{out}(j) > \sum_{j=1}^w M^f_t(j,i') R_{out}(j)$ for $i'$. In particular, this implies that $(M_1^f,\ldots,M_{t-1}^f,M_t',M_{t+1}^f,\ldots,M_{k-1}^f)$ is an optimal solution to Program~\eqref{maxmin_program} that uses budget strictly less than $B$. This contradicts Claim~\ref{clm: maxmin_equalspending}, that shows that the leftover budget can then be used to increase the optimal value of Program~\ref{maxmin_program}, implying that $M^f$ cannot be an optimal solution. Therefore, it must be the case that for all $i \in [w]$, for all $t \in [k-1]$,
\[
\sum_{j=1}^w M_t^f(j,i) R^\top M_{k-1}^f \ldots M_{t+1}^f e_j \geq \sum_{j=1}^w M_t^0(j,i) R^\top M_{k-1}^f \ldots M_{t+1}^f e_j,
\]
or equivalently
\begin{align}\label{eq: increase_utility}
R^\top M_{k-1}^f \ldots M_{t+1}^f M_t^f e_i \geq R^\top M_{k-1}^f \ldots M_{t+1}^f M_t^0 e_i. 
\end{align}
Applying this with $t = 1$, we have that for all starting $q$ on layer $L_1$,
\[
R^\top M_{k-1}^f \ldots M_{1}^f e_q \geq R^\top M_{k-1}^f \ldots M_{2}^f M_1^0 e_q.
\]
Now, suppose by induction that for all $i \in [w]$,
\[
R^\top M_{k-1}^f \ldots M_1^f e_i \geq R^\top M_{k-1}^f \ldots M_{t}^f  M_{t-1}^0 \ldots M_1^0 e_i.
\]
It follows that
\begin{align*}
R^\top M_{k-1}^f \ldots M_1^f e_q
&\geq R^\top M_{k-1}^f \ldots M_{t}^f  M_{t-1}^0 \ldots M_1^0 e_q
\\&= R^\top M_{k-1}^f \ldots M_{t}^f \sum_i \left(M_{t-1}^0 \ldots M_1^0 e_q\right)(i) e_i
\\& = \sum_i \left(M_{t-1}^0 \ldots M_1^0 e_q\right)(i) R^\top M_{k-1}^f \ldots M_{t}^f e_i
\\& \geq \sum_i \left(M_{t-1}^0 \ldots M_1^0 e_q\right)(i) R^\top M_{k-1}^f \ldots M_{t}^0 e_i
\\&=R^\top M_{k-1}^f \ldots M_{t+1}^f  M_t^0 \ldots M_1^0 e_q
\end{align*}
where the second-to-last equation follows from Equation~\eqref{eq: increase_utility}. Therefore, by induction, we have that for all $q$,
\[
R^\top M_{k-1}^f \ldots M_1^f e_q \geq R^\top M_{k-1}^0 \ldots M_1^0 e_q,
\]
directly implying that 
\[
R^\top M_{k-1}^f \ldots M_1^f D_1^0 \geq R^\top M_{k-1}^0 \ldots M_1^0 e_i D_1^0 = W^0.
\]

\end{document}